\documentclass[11pt]{article}

\PassOptionsToPackage{numbers, compress}{natbib}
\usepackage{amsmath,amssymb,amsthm,fullpage,mathrsfs,pgf,tikz,caption,subcaption,mathtools}
\usepackage{natbib}
\usepackage{amsmath,amssymb,amsthm,mathtools,pxfonts}
\usepackage[utf8]{inputenc} 
\usepackage[T1]{fontenc}    
\usepackage{hyperref}       
\usepackage{url}            
\usepackage{booktabs}       
\usepackage{tablefootnote}
\usepackage{amsfonts}       
\usepackage{nicefrac}       
\usepackage{microtype}      
\usepackage{times}
\usepackage[margin=1in]{geometry}
\usepackage{bbm}
\usepackage{enumitem}
\usepackage{xcolor}
\usepackage{cleveref}
\usepackage{bm}

\usepackage{algorithm}
\usepackage{algorithmicx}
\usepackage[noend]{algpseudocode}
\usepackage{booktabs}

\newtheorem{theorem}{Theorem}[section]
\newtheorem{corollary}[theorem]{Corollary}
\newtheorem{proposition}[theorem]{Proposition}
\newtheorem{lemma}[theorem]{Lemma}
\newtheorem{definition}[theorem]{Definition}

\newtheorem{question}[theorem]{Question}

\newtheorem{observation}[theorem]{Observation}
\newtheorem{claim}[theorem]{Claim}
\newtheorem{fact}[theorem]{Fact}

\newcommand{\R}{\mathbb{R}}

\newcommand{\Z}{\mathbb{Z}}

\DeclareMathOperator{\Reg}{Reg}

\DeclareMathOperator{\poly}{poly}
\DeclareMathOperator{\polylog}{polylog}

\newcommand{\eps}{\varepsilon}

\DeclareMathOperator{\NC}{NC}

\newcommand{\HD}[2]{RM^{(#1)}\lp(#2\rp)}
\newcommand{\HDI}[3]{RM^{(#1)}_{#2}\lp(#3\rp)}
\DeclareMathOperator{\TV}{TV}

\newcommand{\E}{\mathbb{E}}

\newcommand{\rbrac}[1]{\left(#1\right)}
\newcommand{\sbrac}[1]{\left[#1\right]}

\newcommand{\vect}{\mathbf}

\newcommand{\lp}{\left}
\newcommand{\rp}{\right}

\newcommand{\GAME}{\{(\Actions_i, R_i)_{i=1}^m\}}

\newcommand{\eqdef}{:=}

\newcommand{\Actions}{A}

\newcommand{\qt}{Q}

\newcommand{\wt}[1]{\widetilde{#1}}
\newcommand{\Lmax}{L_{\max}}

\author{%
  Daniel Beaglehole\thanks{Department of Computer Science and Engineering,
  UCSD, CA 92092.
  Email: \texttt{dbeaglehole@ucsd.edu}}
  \and
  Max Hopkins\thanks{Department of Computer Science and Engineering,
  UCSD, CA 92092. Email: \texttt{nmhopkin@eng.ucsd.edu}. Supported by NSF Award DGE-1650112.}
  \and Daniel Kane\thanks{Department of Computer Science and Engineering / Department of Mathematics, UCSD, California, CA 92092. Email: \texttt{dakane@eng.ucsd.edu}. Supported by NSF Award CCF-1553288 (CAREER) and a Sloan
Research Fellowship.}
  \and
    Sihan Liu\thanks{Department of Computer Science and Engineering,
  UCSD, CA 92092.
  Email: \texttt{sil046@ucsd.edu}
  }
    \and
    Shachar Lovett\thanks{Department of Computer Science and Engineering, UCSD, CA 92092. Email: \texttt{slovett@cs.ucsd.edu}. Supported by NSF Award CCF-1953928.}
}
\date{}
\title{Sampling Equilibria: Fast No-Regret Learning in Structured Games}

\begin{document}

\maketitle
\begin{abstract}
Learning and equilibrium computation in games are fundamental problems across computer science and economics, with applications ranging from politics to machine learning. Much of the work in this area revolves around a simple algorithm termed \emph{randomized weighted majority} (RWM), also known as ``Hedge'' or ``Multiplicative Weights Update,'' which is well known to achieve statistically optimal rates in adversarial settings (Littlestone and Warmuth ’94, Freund and Schapire ’99). Unfortunately, RWM comes with an inherent computational barrier: it requires maintaining and sampling from a distribution over all possible actions. In typical settings of interest the action space is exponentially large, seemingly rendering RWM useless in practice.

In this work, we refute this notion for a broad variety of \emph{structured} games, showing it is possible to efficiently (approximately) sample the action space in RWM in \emph{polylogarithmic} time. This gives the first efficient no-regret algorithms for problems such as the \emph{(discrete) Colonel Blotto game}, \emph{matroid congestion}, \emph{matroid security}, and basic \emph{dueling games}. As an immediate corollary, we give a polylogarithmic time meta-algorithm to compute approximate Nash Equilibria for these games that is exponentially faster than prior methods in several important settings. Further, our algorithm is the first to efficiently compute equilibria for more involved variants of these games with general sums, more than two players, and, for Colonel Blotto, multiple resource types.

\end{abstract}

\newpage
\tableofcontents
\newpage

\section{Introduction}\label{sec:intro}
Online learning and equilibrium computation in games has long played a major role in our understanding of human behavior and general multi-agent systems, with applications ranging all the way from politics \cite{behnezhad2019optimal,kovenock_roberson_2012} and national defense \cite{tambe2011security} to complexity theory \cite{daskalakis2009complexity,rubinstein2016settling} and machine learning \cite{freund_schapire_1996,freund_schapire_1999,andoni2021learning}. Perhaps the most celebrated line of work in this area is the introduction and analysis of randomized weighted-majority (RWM) and its `mixed' variant (Optimistic) Hedge \cite{littlestone1994weighted,freund_schapire_1996, rakhlin2013optimization,daskalakis2021near}. These powerful algorithms allow players to engage in repeated gameplay \emph{without regret}, in the sense that the overall loss experienced by any player is not much more than that of the best fixed strategy, even against an arbitrary, adaptive adversary. Such a guarantee is not only powerful in its own right, but is also known to converge quickly to equilibria when performed by all players in repeated rounds of play \cite{cesa2006prediction}.

Randomized weighted majority is a surprisingly simple algorithm given its powerful guarantees.
In each round of a repeated game, a player following RWM samples a strategy $s$ with probability proportional to its (exponentiated) historical loss $\ell(s)$:
\begin{equation}\label{intro:eq-RWM}
\Pr[\text{Player chooses $s$}] \propto \beta^{\ell(s)}
\end{equation}
for some specified `learning rate' $\beta \in (0,1)$. RWM is also well studied in the setting where the player `plays' the distribution itself (typically called a mixed strategy), and experiences its expected loss. This variant, called \textit{Hedge}, is perhaps the best studied algorithm in all of learning in games \cite{freund_schapire_1996,cesa2006prediction}.

Unfortunately, while RWM and Hedge are \emph{statistically} optimal \cite{littlestone1994weighted,freund_schapire_1996}, they come with an inherent computational barrier: both techniques crucially rely on tracking a distribution over all possible actions. Since the number of actions is typically exponential in the relevant parameters of the game (e.g.\ in the famous Colonel Blotto problem), this seems to render both Hedge and RWM completely infeasible. 

It turns out, however, that this intuition is not entirely correct. In many important settings the distributions that arise from RWM are \textit{highly structured}, and while it still may not be possible to efficiently \textit{output} the distribution itself as in Hedge, it is sometimes possible to efficiently \textit{sample} from it. It is known, for instance, that RWM can be implemented in \textit{polylogarithmic} time when actions are given by the $k$-edges of a complete hypergraph and rewards decompose linearly over vertices \cite{warmuth2008randomized}. This raises an important question:
\begin{center}
    \emph{When is it possible to efficiently sample in Randomized Weighted Majority?}
\end{center}

Toward this end, we introduce a natural generalization of the complete hypergraph setting we call \emph{linear hypergraph games}, where actions are given by $k$-edges of an \textit{arbitrary} hypergraph, and the reward of any edge similarly decomposes as a sum over individual reward functions on its vertices (see \Cref{sec:linear_hypergraph_game} for more detail). This simple definition captures a surprising number of settings studied in the literature including resource allocation problems like Colonel Blotto \cite{borel_1953}, along with other widely-studied settings such as congestion \cite{rosenthal1973class}, security \cite{tambe2011security,ahmadinejad_dehghani_hajiaghayi_lucier_mahini_seddighin_2019,szeszler2017security,baiou2019faster}, and basic dueling games \cite{immorlica2011dueling,ahmadinejad_dehghani_hajiaghayi_lucier_mahini_seddighin_2019}.

In this work, we show it is indeed possible to efficiently (approximately) sample from RWM over several important subclasses of linear hypergraph games including Colonel Blotto and its variants, matroid congestion \cite{vocking2006congestion}, matroid security \cite{tambe2011security,ahmadinejad_dehghani_hajiaghayi_lucier_mahini_seddighin_2019,szeszler2017security,baiou2019faster}, and basic dueling games \cite{immorlica2011dueling,ahmadinejad_dehghani_hajiaghayi_lucier_mahini_seddighin_2019}. This leads to the first algorithms for no-regret learning in these settings that are \emph{polylogarithmic} in the size of the state space, and thereby the first polylogarithmic time algorithms for (approximate) equilibrium computation. On top of giving an exponential improvement over prior results, this also constitutes the first efficient algorithm for equilibrium computation whatsoever in several more involved settings such as dice games, Colonel Blotto with multiple resources, and for multiplayer and general-sum variants of all games we consider. 
These results are (informally) summarized in \Cref{tab:runtimes}.

Our techniques are largely based on two main paradigms: dynamic programming, and Monte Carlo Markov Chains (MCMC). Generalizing seminal work on learning $k$-sets and other structured concepts \cite{warmuth2008randomized,koolen2010hedging}, we show that the distributions arising from RWM on linear hypergraph games correspond to well-studied structure in approximate sampling and statistical physics called \emph{external fields}. In resource allocation games like Colonel Blotto that are played over (ordered) fixed-size partitions of $n$, we exploit this structure to build a dynamic program that approximately computes the normalization factor of \Cref{intro:eq-RWM} (often called the partition function). On the other hand, in settings like matroid congestion and security, we rely on deep results from the MCMC-sampling literature showing that any hypergraph that is a sufficiently good high dimensional expander can be sampled under arbitrary external fields \cite{anari2019log,alimohammadi2021fractionally,anari2021entropic}. To the authors' knowledge, these are the first applications of approximate sampling techniques to game theory.

\begin{table}[h!]
    \centering
    \begin{tabular}{ccc}
        \toprule
        Game & Known runtime (exact) & Our runtime (apx.) \\
        \midrule
        (Discrete) Colonel Blotto & $n^{13}k^{14}$ \cite{ahmadinejad_dehghani_hajiaghayi_lucier_mahini_seddighin_2019} & $k^4\log^4(n)$ \\
        Dice & NA & $k^4\log^4(n)$ \\
        Matroid Security & $n^4k^2$ \cite{baiou2019faster} & $k^3\log^3(n)$\\ 
        Matroid Congestion & $n k$ \cite{ackermann2008impact} & $k^4 \log^3(n)$\\ 
        Ranking Duel  & $n^{12}$ \cite{immorlica2011dueling} \tablefootnote{This work did not give a runtime for their linear program. Here we naïvely apply the bound for ellipsoid method from \cite{magen_2005}. The work of \cite{ahmadinejad_dehghani_hajiaghayi_lucier_mahini_seddighin_2019} can also be applied in this setting, but is largely geared towards being more general than more efficient (and also does not list a runtime).} & $n^9$\\ 
        \bottomrule
    \end{tabular}

    \caption{Rough asymptotic runtimes for equilibrium computation in contrast to previous methods with two players. Approximation factor and maximum reward are set to $O(1)$ along with some other game specific parameters (See \Cref{def:collision-sensitive} for details).}
    \label{tab:runtimes}
\end{table}

\subsection{Results}

We briefly review the theory of games, equilibria, and no-regret learning before discussing our results in more formality. Games are mathematical objects that model (possibly non-cooperative) interaction between rational agents. A (simultaneous) game consists of a set of actions $A_i$ for each player, and reward functions $R_i$ mapping action tuples to rewards (real numbers). Players seek to maximize their own reward, and optimal play is typically characterized by Nash equilibria: randomized strategies such that no player can improve by deviating. By the historic result of \cite{nash1951non}, every finite game has at least one NE. As they are not always efficiently computable \cite{daskalakis_goldberg_papadimitriou_2006}, one often instead hopes to understand weaker notions such as Coarse Correlated Equilibria (CCE), where the strategies of different players are chosen in coordination with one another (see \Cref{sec:game}).

There is a deep connection between equilibrium computation and no-regret learning in games. We consider the typical adaptive online setting in which, in each round, a learner chooses an action and receives an adversarially selected loss that may depend on the learner's previous actions (see \cite[Chapter 4]{cesa2006prediction}). An algorithm is said to have ``no-regret'' when the expected loss suffered by the learner in sequential rounds grows sublinearly compared to the loss of the best fixed action in hindsight. No-regret learning is itself a powerful tool, as it allows for optimal play against sub-optimal opponents (unlike equilibria which only model the setting where all agents play optimally). Furthermore, it is well known that any no-regret algorithm\footnote{Formally we may need to require that a players strategy depends only on the opponents history and not their own \cite{cesa2006prediction}. This is satisfied by all algorithms considered in this work.} leads to approximate equilibrium computation with similar runtime simply by simulating the algorithm for all players for sufficiently many rounds. RWM, for instance, is well-known to satisfy the following (optimal) regret guarantee.
\begin{lemma}[RWM is No-Regret {\cite[Lemma 4.1]{cesa2006prediction}}]
The regret of RWM over $T$ rounds, $N$ actions, and with rewards in $[-\Lmax,\Lmax]$, satisfies
\[
\Reg_T \leq O\left(\Lmax \sqrt{T\log(N)}\right),
\]
against any adaptive adversary (with high probability).
\end{lemma}
In fact, it is important to note in our setting that essentially  all guarantees of RWM also hold in the approximate regime, where the learner only $\delta$-approximately samples from the distribution in \Cref{intro:eq-RWM} in each round (in Total Variation distance). We call such algorithms $\delta$-approximate RWM ($\delta$-RWM). It is not hard to show that $\delta$-RWM also satisfies the above regret guarantees for small enough $\delta$ (see \Cref{lem:apx-mwu}). We now cover four of the main settings in which we give new algorithms for no-regret learning and equilibrium computation through efficient implementatin of $\delta$-RWM: Colonel Blotto, Matroid Security, Matroid Congestion, and Dueling games. We note that all results are given in the algebraic computation model for simplicity (where algebraic operations such as addition and subtraction are considered to be unit time), but can easily be moved to the standard bit model with no substantial loss in running time (see \Cref{app:bit}).

\subsubsection{The Colonel Blotto Game}
The Colonel Blotto game was originally described by Borel in 1921 \cite{borel_1953} and formalizes how warring colonels should distribute soldiers over different battlefields. In the most general version of this game, two colonels have $n_1$ and $n_2$ soldiers that they must assign to $k$ different battlefields, each with a non-negative integer weight. A colonel wins a battle (receiving its weight in reward) if they assign more armies to that battle than their opponent. Each colonel seeks to maximize the total weight of battles won in a single assignment.

Despite its breadth of applications and simplicity to state, the first polynomial time algorithm to compute optimal strategies for this game was only developed recently in \cite{ahmadinejad_dehghani_hajiaghayi_lucier_mahini_seddighin_2019}. This breakthrough result deservedly received significant media attention \cite{Insider, mewright_2016}, but struggled to see any practical use due to an infeasible $O(n^{13}k^{14})$ running time (where $n=\max\{n_1,n_2\}$). To this day, this is the only known algorithm to provably compute exact optimal strategies for the (discrete) Colonel Blotto game with arbitrary parameters in polynomial time. Though some progress has been made towards more practical algorithms in different settings \cite{behnezhad_dehghani_derakhshan_hajiaghayi_seddighin_2017}, even these methods cannot handle parameters beyond a few hundred troops \cite{vu_loiseau_silva_2018}.


Indeed, solving the Colonel Blotto problem is now only more relevant than it was in 1921, with practical applications in a large swath of market competitions including advertising and auctions \cite{roberson_2006}, budget allocation \cite{KVASOV2007738}, elections \cite{laslier_picard_2002}, and even ecological modeling \cite{golman2009}. 
We give the first no-regret learning algorithm for the Colonel Blotto games under the most general setting \cite{kovenock2021generalizations}, where rewards are heterogeneous across battles and players and different players are allowed different troop capacities. Moreover, our algorithm runs in time \textit{polylogarithmic} in the state space, making it extremely efficient in the regime where $n \gg k$ (i.e.\ there are many more troops than battlefields).

\begin{theorem}[Blotto without Regret (Informal \Cref{thm:cb-no-regret})]\label{intro:Blotto1}
In a Colonel Blotto game, for a player with $n$ soldiers, $k$ battlefields, and maximum reward bounded by $\Lmax$, $\delta$-RWM can be implemented over $T$ rounds of play in time: 
\[
\wt{O} \lp(  T^{3}\Lmax k\log(n)\delta^{-1}
\rp),
\] and is no-regret. In the regime where $n=O(k^2)$, we give a faster algorithm running in time $\wt O(Tnk)$.
\end{theorem}
\Cref{intro:Blotto1} is the first no-regret algorithm for Colonel Blotto in online adaptive settings, and also gives the fastest known algorithms to compute (approximate) Nash equilibria provided the game is zero-sum, and approximate coarse correlated equilibrium in general sum settings with many players. We state the theorem for the two-player zero-sum setting here.

\begin{corollary}[Equilibrium Computation for Blotto  (informal \Cref{cor:blotto-ec})]\label{intro:Blotto2}
Let $n = \max(n_1, n_2)$, where $n_1, n_2$ are the soldier counts for the two
player Colonel Blotto game. Let $\Lmax$ be maximum reward of the game. There exists
an algorithm to compute an $\eps$-approximate Nash equilibrium for the two-player Colonel Blotto Game in time
$$
\wt O\lp( \Lmax^7k^4\log^4(n) \eps^{-6}
\rp)
$$ 
with high probability. When $n=O(k^2)$, we give a faster algorithm running in time $\wt O( n k^2 \Lmax^2 \eps^{-2})$.
\end{corollary}
Not only is this algorithm exponentially faster than any prior work in most relevant scenarios (namely when $n \gg k$), it is also the first known method for computing CCE for \textit{multiplayer} Blotto at all. Even more generally, our algorithm extends to a number of other variants of Blotto (or `resource allocation' problems) such as Dice games and settings with multiple types of troops known as the \textit{Multi-resource Colonel Blotto} problem \cite{behnezhad_dehghani_derakhshan_hajiaghayi_seddighin_2017} (though in this latter setting we lose the logarithmic dependence on $n$). We cover these further applications in \Cref{sec:dp-applications} and \Cref{sec:multi-resource}.

\subsubsection{Congestion Games}
Another natural example is a \textit{congestion game}, a class introduced by Rosenthal \cite{rosenthal1973class} to model resource competition among greedy players. In a congestion game, $m$ players compete to select from a set of $n$ resources and receive rewards depending on how many players chose a particular resource. Classical examples of congestion games include routing traffic (pick the least congested route) and variants of the famed El Farol Bar Problem \cite{arthur1994inductive} (players aim to choose a bar that is neither too under nor over-crowded).

Unlike Blotto, equilibrium computation is known to be hard for congestion games, namely (PPAD $\cap$ PLS)-complete \cite{babichenko2021settling}. However, this can be circumvented when the underlying strategy spaces are sufficiently combinatorially structured. It has long been known, for instance, that a Nash equilibrium can be found in time $\wt O(m^2nq k)$ via iterated best-response when all strategies are given by the bases of a rank-$k$ matroid\footnote{Matroid bases can be thought of as a generalization of the combinatorial properties enjoyed by spanning trees, see \Cref{sec:matroid} for details.} over $n$ resources of $q$ types \cite{ackermann2008impact}. We show matroid congestion games are similarly well-behaved under RWM, and provide a near-optimal no-regret algorithm in both a computational and statistical sense.
\begin{theorem}[Congestion without Regret (informal \Cref{cor:congestion-nr})]\label{intro:Congestion1}
Let $\mathcal{I}=\{ \{A_i\}_{i=1}^m,c\}$ be a congestion game over a size-$n$ ground set $\Omega$ with $q$ resource types where each $A_i$ is the set of bases of a rank-$k$ matroid. Then $\delta$-RWM can be implemented for $T$ rounds in time
\[\wt O\left(kT\log(n)\log(\delta^{-1})(q+kmT)\right),
\] 
and is no-regret.
\end{theorem}

To our knowledge, this is the first efficient no-regret algorithm for matroid congestion. Moreover, in the setting where there are $\text{polylog}(n)$ resource types, the algorithm leads to exponentially faster (approximate) equilibrium computation than the typical best response strategy (albeit for CCE rather than Nash).
\begin{corollary}[Equilibrium Computation for Congestion (informal \Cref{cor:congestion-ec})]\label{intro:Congestion2}
Let $\mathcal{I}=\GAME$ be a congestion game over a size-$n$ ground set $\Omega$ with $q$ resource types where each $A_i$ is the set of bases of a rank-$k$ matroid. Then there exists an algorithm to compute an $\varepsilon$-CCE in time
\[
\wt O\left(m^2\Lmax^4k^4\log^3(n)\eps^{-4}+
qm\Lmax^2k^2\log^2(n)\eps^{-2}\right)
\]
with high probability.
\end{corollary}

\subsubsection{Security Games}
While slightly less intuitive, games modeling security also fit within the resource allocation paradigm. \textit{Security games} are a basic two-player setting modeling the behavior of a limited-resource player defending $n$ targets, and an adversarial attacker. Each target in the game has a cost to defend, and a ``$k$-resource'' defender may choose a (possibly restricted) $k$-set to defend. Similarly, each target has a cost to attack, and the attacker chooses a single element, receiving a reward depending on whether or not the selected target was defended by the opponent. Depending on the cost/reward structure, security games model several real-world scenarios, ranging from allocating defensive resources at military checkpoints to choosing a path to transmit critical resources (in the latter the attacker actually \textit{wins} if they attack a `defended' node). Indeed, security games have actually seen significant use in critical real-life infrastructure such as checkpoint placements at LAX and US Coast Guard and Federal Air Marshal Service patrol schedules \cite{tambe2011security}.

Given their practical importance, it is no surprise equilibrium computation is well-studied in the security game setting \cite{tambe2011security,ahmadinejad_dehghani_hajiaghayi_lucier_mahini_seddighin_2019,szeszler2017security,baiou2019faster}, and polynomial time algorithms are known in several settings, notably including when allocation constraints are given by matroid bases \cite{szeszler2017security,baiou2019faster}. Unfortunately, as is the case in previous work on Blotto, known algorithms are not practically useful and have large polynomial factors in the number of targets. We take a major step toward resolving this issue by showing $\delta$-RMW can be implemented in time \textit{polylogarithmic} in $n$, an exponential improvement over prior techniques \cite{tambe2011security,szeszler2017security,baiou2019faster}.
\begin{theorem}[Security without Regret (Informal \Cref{cor:sec-nr})]\label{intro:Security1}
Let $\mathcal{I}$ a security game over the bases of a rank-$k$ matroid over $n$ targets with $q$ distinct defender costs. Then $\delta$-RMW can be implemented for $T$ rounds in time
\[
\wt O(kT\log(n)\log(\delta^{-1})(q+T)),
\] 
and is no-regret.
\end{theorem}

\begin{corollary}[Equilibrium Computation for Security (Informal \Cref{cor:sec-ec})]\label{intro:Security2}
Let $\mathcal{I}$ be a security game over the bases of a rank-$k$ matroid over $n$ targets with $q$ distinct attacker and defender costs. Then it is possible to compute an $\varepsilon$-CCE in time
\[
\wt O\left(\Lmax^4k^3\log^3(n)\varepsilon^{-4}+q\Lmax^2k^2\log^2(n)\eps^{-2}\right).
\]
If the game is zero-sum, the resulting strategy is $\varepsilon$-Nash.
\end{corollary}

\subsubsection{Dueling Games}
Finally, we make a slight departure from the resource allocation framework to consider the popular class of \textit{dueling games} studied in \cite{immorlica2011dueling,ahmadinejad_dehghani_hajiaghayi_lucier_mahini_seddighin_2019}. Dueling games model competitive optimization between two players over a randomized set of events. We will focus our attention on one of the simplest dueling games called \textit{ranking duel} (also known as the `Search Engine game') where two players compete to rank $n$ elements over a known distribution $\mu$, and win a round if they rank $x \sim \mu$ higher than the opponent. This classically models the problem of search engines competing to optimize a page ranking given a known distribution over searches.

Equilibrium computation is well-studied in dueling games, and algorithms are known in a few settings via a mix of bilinear embedding techniques and reduction to non-competitive optimization \cite{immorlica2011dueling, ahmadinejad_dehghani_hajiaghayi_lucier_mahini_seddighin_2019}. As in previous settings, however, the algorithms are too slow to be of practical use. 
In contrast, we focus our attention only on the basic ranking duel, but give both a faster algorithm and a novel no-regret guarantee over the original space.
\begin{theorem}[Dueling without Regret (Informal \Cref{cor:duel-nr})]\label{intro:Dueling1}
Let $\mathcal{I}$ be an instance of ranking duel. Then $\delta$-RMW can be implemented for $T$ rounds in time 
\[
\wt O(T^2n^7\log(\delta^{-1})),
\] 
and is no-regret.
\end{theorem}

\begin{corollary}[Equilibrium Computation for Ranking Duel (Informal \Cref{cor:duel-ec})]\label{intro:Dueling2}
Let $\mathcal{I}$ be an instance of ranking duel. Then it is possible to compute an $\varepsilon$-CCE in time 
\[
\wt O(n^9\varepsilon^{-4}).
\]
If the game is zero-sum, the resulting strategy is $\varepsilon$-Nash.
\end{corollary}
While a running time of $O(n^9)$ can hardly be claimed as practical, the broader technique used in this result has the eventual possibility of running in near-linear time. We discuss this further in \Cref{sec:discussion}.
\subsection{Techniques}
At their core, our results all stem from the ability to approximately sample distributions arising from randomized weighted majority on various linear hypergraph games. Recall that RWM maintains a mixed strategy, which we denote as the RWM distribution, whose probabilities are proportional to their (exponentiated) total historical loss (negative reward):
\begin{equation}\label{eq:RWM}
\forall x \in A_i: \Pr(x) \propto \beta^{\ell^T(x)}
\end{equation}
where $\ell^T(x)$ is the total loss experienced by pure strategy $x$ up to round $T$. As discussed earlier in the section, we typically cannot hope to maintain this distribution explicitly, but it may still be possible to sample from it in polylogarithmic time. Furthermore, while sampling such a distribution exactly is a challenging task (and very few such algorithms are known), \emph{approximate} sampling is perfectly sufficient in our setting. Indeed, our approximate variant $\delta$-RWM satisfies essentially the same guarantees as RWM itself.
\begin{lemma}[$\delta$-RWM is No-Regret (\Cref{lem:apx-mwu})]\label{intro:no-regret}
$\delta$-RWM over $N$ actions has 
\[
Reg(T) \leq O \rbrac{\Lmax \sqrt{T \log N} + \delta \Lmax T}
\]
expected regret, where $\Lmax$ is the maximum loss experienced by any action.
\end{lemma}
As discussed at the start of Section~\ref{sec:intro}, no-regret algorithms like RWM are classically used to compute equilibria of the base game by simulating repeated play across all players. While much of the current literature centers around the Hedge algorithm that `plays' an entire mixed strategy in each round, classical (and therefore approximate) RWM still leads to equilibrium computation with high probability, by the classic result that no-regret implies equilibrium computation \cite{freund_schapire_1999}.

\begin{lemma}[Approximate RWM $\to$ Equilibria (Informal \Cref{cor:compute-cce})]\label{intro:CCE}
Let $\mathcal{I}$ be an $m$-player game where each player has at most $N$ strategies. Let $\{(x_1^{(t)},\ldots,x_m^{(t)})\}_{t=1}^T$ be the strategies arising from $T$ rounds of $\delta$-approximate RWM. There exist universal constants $C > 0$ such that for $T= C \cdot {\Lmax^2\varepsilon^{-2} \cdot \log(N)}$ rounds, and approximation parameter $\delta \leq {\eps}/{(C\Lmax)}$, these strategies constitute an $\varepsilon$-CCE with high probability (Nash if the game is two-player zero-sum).
\end{lemma}
As a result, efficient no-regret learning and equilibrium computation truly reduces to the existence of an efficient approximate sampling scheme for distributions arising in the execution of ($\delta$-approximate) RWM. Of course, this is easier said than done. While approximate sampling is easier than its exact variant, it is still a challenging problem, even over structured domains. Using a mixture of novel sampling techniques and reductions to known methods in the literature, we show it is indeed possible to efficiently sample from RWM across a wide variety of structured games. Our strategies fall into two main paradigms: dynamic programming (DP), and Monte Carlo Markov Chains (MCMC). 

\subsubsection{Sampling via Dynamic Programming}
We start with the former: sampling in basic resource allocation settings via dynamic programming. At its most general, resource allocation problems are played over (possibly constrained) fixed size partitions of $n$. The discrete Colonel Blotto game on $n$ troops and $k$ battlefields is the simplest example of this problem, where the strategy space corresponds to the set of all $k$-size (ordered) partitions of $n$ (i.e.\ assignments $x_1,\ldots,x_k$ such that $\sum x_i = n$). In this section, we will focus only on the Colonel Blotto problem---general resource allocation follows from very similar arguments (see \Cref{sec:dp-sample} for details).

Our goal is now to design an algorithm for approximately sampling distributions over strategies of the Colonel Blotto game that arise from RWM. In this setting, it will actually be easier to solve an equivalent problem, computing the normalizing factor of \Cref{eq:RWM}, otherwise known as the \textit{partition function}:
\[
f_k(n) = \sum\limits_{x_1+\ldots+x_k=n}\beta^{\ell(x)} = \sum\limits_{x_1+\ldots+x_k=n}\prod\limits_{h=1}^k \beta^{\ell_h(x_h)},
\]
where $\ell_h(x_h)$ is the historical losses from the $h$-th battlefield over previous rounds of play if one were to place $x_h$ soldiers on that battlefield. Notice that once we know the value of $f_{k'}(n')$ for all $k'\leq k$ and $n'\leq n$, it is actually possible to exactly sample from \Cref{eq:RWM} (and therefore implement RWM). In particular, one does this simply by sequentially sampling the number of troops to put in each battlefield conditional on prior choices in the following manner:
\begin{align*}
\Pr[x_1 = y] &\propto
\beta^{ \ell_1(y) } \cdot f_{k-1}(n - y) \text{ for the first battlefield, }\\
\Pr \lp[x_{h+1} = y|  x_{1 \cdots h} \rp] 
&\propto 
\beta^{ \ell_{h+1}(y) } \cdot f_{k-h-1}\lp(n -
\lp( \sum_{j=1}^h x_j \rp) - y\rp) \text{ for the remaining battlefields}.
\end{align*}
One can easily check the joint distribution arising from this procedure is exactly the RWM distribution.

Thus we have reduced our problem to computing the partition functions $f_{k'}(n')$. This can be done by a simple dynamic programming argument, and in particular by noticing that:
\begin{align}
    f_{k'}(n') = \sum_{i=0}^{n'} \beta^{\ell_{k'}(i)}
    \cdot f_{k'-1} \lp( n' - i \rp).
\end{align}
Since filling each entry $f_{k'}(n')$ takes time at most $O(n)$ given that $f_{h-1}$ is pre-computed, we can fill the entire DP table in time $O(n^2 k)$.\footnote{We note that this can actually be improved to near-linear in $n$ using the Fast Fourier Transform.}

While this procedure already gives the first no-regret learning algorithm for Blotto in the adversarial setting and by far the fastest known equilibrium computation, one can still hope to do much better. Indeed, it is known that there exist $\varepsilon$-Nash Equilibrium with support that is \emph{logarithmic} in the size of the state-space \cite{lipton1994simple}, so there is hope in building a \emph{polylogarithmic} time algorithm (equivalently, a polynomial time algorithm in the description complexity of the equilibria). We show this is indeed possible by building an approximation scheme for the above DP. The key is to observe that the partition functions $f_{k'}(n')$ are bounded and monotonic. Roughly speaking, this means $f$ can be approximated within multiplicative $(1 \pm \varepsilon)$ factors by a piece-wise function with only $\text{poly}({k\log(n)}/{\varepsilon})$ pieces (which is polylogarithmic in the size of the state space). 

By carefully computing and maintaining approximate versions of the partition function, we can run a modified variant of the same dynamic program that computes approximations for all $nk$ partition functions $f_{k'}(n')$ (despite their sizes, these can indeed be presented in only $\text{poly}(k\log(n))$ bits due to being piece-wise). Once we have approximately computed the partition functions, it is easy to show that a similar sampling scheme as discussed for the exact case gives an efficient approximate sampling scheme for RWM running in $\text{poly}({k\log(n)}/{\varepsilon})$ time. Combined with \Cref{intro:no-regret} and \Cref{intro:CCE}, this results in the first polylogarithmic time algorithm for no-regret learning and (approximate) equilibrium computation for Colonel Blotto (\Cref{intro:Blotto1} and \Cref{intro:Blotto2}), as well as for a number of related resource allocation variants discussed later in the paper (e.g.\ multi-resource Blotto and Dice games).

\subsubsection{MCMC-methods}
While dynamic programming is a powerful algorithmic method for structured computation, there are many combinatorial settings common to games we cannot hope to handle via such techniques. Building an analogous exact-counting based DP for games over bipartite matchings or matroids, for instance, would give efficient algorithms for classical \#P-hard problems such as the permanent and counting matroid bases \cite{colbourn1995complexity}. On the other hand, we do actually know of approximation algorithms for these problems based on a powerful tool called MCMC-sampling \cite{jerrum2004polynomial,anari2019log}.

MCMC-sampling is an elegant method for approximately sampling from a distribution $\mu$ with exponential size support usually traced back to Ulam and Von Neumann in the 1940s (see e.g.\ \cite{eckhardt1987stan}). The idea is simple. Imagine we can construct a Markov chain (random process) $M$ satisfying the following three conditions:
\begin{enumerate}
    \item The stationary distribution of $M$ is $\mu$
    \item A single step of $M$ can be implemented efficiently
    \item $M$ converges quickly to its stationary distribution.
\end{enumerate}
Approximate sampling would then simply boil down to finding a starting configuration and running the chain until it is within $\delta$ of stationary (this typically takes around $O(\log(N/\delta))$ samples for a good chain).

Unsurprisingly, while MCMC-sampling itself is a simple technique, the design and analysis of Markov chains is a difficult task, and general recipes for their construction are known in very few scenarios. One particularly well-studied setting in the literature that arises from simulation problems in statistical physics are \textit{external fields}. Given a hypergraph $\Omega \subset {[n] \choose k}$, the distribution arising from external field $w \in \R_+^n$ simply assigns each $k$-set a probability proportional to the product of its fields:
\[
\Omega^w(s) \ \propto \ \prod\limits_{v \in s} w(v).
\]
External fields often correspond to particularly natural problems, and are well-studied in the literature. In a recent breakthrough series of works, for instance, it was shown that approximate sampling under external fields is possible whenever the underlying state-space is a good enough high dimensional expander \cite{kaufman2020high,alev2020improved,anari2019log,anari2021entropic}.\footnote{More formally, when the space satisfies a property known as `fractional log-concavity.'}

This is particularly relevant to our setting since it is a simple observation that the distributions arising from RWM on a linear hypergraph game are exactly given by the application of an external field over the action space.
\begin{observation}[RWM $\to$ External Fields (Informal \Cref{obs:hedge-to-external})]
Let $\mathcal{I}=\GAME$ be an $m$-player linear hypergraph game. Then for any $A_i$ and any round of play, Player $i$'s RWM distribution can be written as the application of an external field $w$ to $A_i$.
\end{observation}
As a result, no-regret learning and equilibrium computation are possible in any linear hypergraph game whose state space can be sampled under arbitrary external fields. As a short detour, it is worth noting that the result in the previous section can be phrased as a slight refinement of this statement. At a technical level, our resource allocation algorithm simply corresponds to an efficient approximate sampling scheme for fixed-size partitions of $n$ under \textit{monotonic} external fields (corresponding to the fact that assigning more troops to a battlefield always results in at least as many victories).

Many well-studied games in the literature have state spaces where efficient approximate sampling schemes under external fields exist. In this work we focus mostly on games played on matroids (e.g.\ matroid congestion, security), and dueling games arising from bipartite matchings such as ranking duel. Both settings have well-known sampling schemes over external fields \cite{anari2019log,cryan2019modified,jerrum2004polynomial}, which leads to our results for Congestion, Security, and Dueling games (Theorems \ref{intro:Congestion1},\ref{intro:Security1}, \ref{intro:Dueling1} and Corollaries \ref{intro:Congestion2},\ref{intro:Security2}, \ref{intro:Dueling2} respectively).

\subsection{Discussion}\label{sec:discussion}
In this work, we present two potential paradigms for learning in games via approximate sampling. In this section we touch on the pros and cons of each method, their likelihood to generalize beyond the settings considered in this work, and natural open problems.

\subsubsection{Dynamic Programming vs MCMC-sampling}
Broadly speaking, the DP and MCMC approaches we develop in this work seem to be largely incomparable. Dynamic programming works well in relatively unconstrained resource allocation problems, where recursive structure allows for inductive computation of the partition function. On the other hand, typical MCMC methods (which are usually \textit{local} in nature) actually fail drastically in this sort of setting due to the need for global coordination. One natural example of this issue appears in the Colonel Blotto game. Imagine a scenario where Colonel $A$ has $k$ more troops than Colonel $B$, then there always exists a configuration where $A$ wins every battle by assigning one more troop in each battlefield than $B$. Finding this sort of optimum, however, requires coordinated planning across battlefields. Typical MCMC methods like Glauber dynamics (see \Cref{sec:mcmc-sample}) only look at a few battlefields at a time, and therefore struggle to converge to such solutions. Simulations confirm this intuition---even for small $n$ and $k$ Glauber dynamics seem to exhibit very poor mixing on distributions arising in RWM.

On the other hand, as we mentioned in the previous section, our dynamic programming approach has a significant issue in any setting with non-trivial combinatorial structure. In particular, because the underlying method relies on exactly computing the partition function, constructing any such method for a problem like matroid games is \#P-hard. On the other hand, local chains such as the Glauber Dynamics mix extremely fast in these settings, providing near-optimal algorithms.

Of course, neither of these arguments rules out either approach. It is possible there exist successful MCMC methods for Blotto that are more \textit{global} in nature---indeed the main insight leading to the resolution of approximate permanent was exactly such a Markov chain that avoided these issues \cite{jerrum2004polynomial}. On the other hand, there may exist DP-based approaches that do not go through computing the partition function. Understanding in which scenarios these two or other potential sampling methods may apply remains an interesting and important open problem if we wish to extend efficient learnability in games beyond the few structured settings considered in this work.

\subsubsection{Further Open Problems}
Our work gives the first no-regret learning guarantees and polylogarithmic equilibrium computation for several well-studied settings in game theory, but there is still much to be done. Perhaps the most obvious open directions involve improving the computational efficiency (and therefore practicality) of our algorithms. The polynomial dependencies of our algorithm would be universally improved if we can show that the $\delta$-approximate \textit{optimistic} variant of RWM achieves $\wt{O}(1)$ regret in games like its deterministic counter-part, Optimistic Hedge \cite{daskalakis2021near}, even for polynomial approximation $\delta$.

\begin{question}[Optimistic-RWM]
Does $\delta$-RWM with weights $\{w_i^{(t)}\}_{i,t}$ and optimistic updates \\${w_i^{(t+1)} \leftarrow w_i^{(t)} \cdot \beta^{2\ell^{(t)}_i - \ell^{(t-1)}_i}}$ achieve $\polylog(T)$ regret in games (even for $\delta = (\polylog N)^{-1}$)?
\end{question}

These techniques are well known to give a substantial improvement in the exact setting \cite{daskalakis2021near}, but their analysis is subtle and may be nontrivial to adapt to the $\delta$-approximate sampling variant we need for efficient computation.

Similarly, our algorithm for dueling games (while faster than prior work in the worst-case), is not practical at $O(n^9)$ running time. One interesting question is whether the MCMC-sampling technique can be improved in this setting using the fact that the weights arising from RWM are not arbitrary, but exhibit \textit{monotonic} structure (in the sense that ranking a page higher is always better). This actually corresponds to a well-studied problem in the sampling and geometry of polynomials literature (monotone permanent \cite{branden2011proof}), but giving an improved sampling algorithm over the JSV-chain \cite{jerrum2004polynomial} remains an interesting open problem.

\begin{question}[Sampling with Monotone Weights]
Can a perfect matching in a complete bipartite graph with $n$ nodes under monotone external fields be sampled in faster than $O(n^7)$ time? In near-linear time?
\end{question}
There is certainly hope in this direction, as recent years have seen many breakthroughs towards near-optimal MCMC methods, including similar linear time guarantees on problems that once seemed infeasible \cite{anari2019log,cryan2019modified,chen2021optimal}.

Another natural direction is to try to strengthen the type of equilibria we compute in multiplayer and general-sum games. Foremost in this direction are the so-called Correlated Equilibria (CE), a substantially stronger notion than CCE which allows a player to switch strategies even after they receive instructions from the coordinator. It was recently shown that a variant of Optimistic Hedge converges quickly to CE in multiplayer, general-sum games \cite{anagnostides2022near}.  It is an open question whether an approximate, sampled variant could do the same. We pose this as the following,

\begin{question}[Correlated Equilibria with RWM]
Is there a variant of $\delta$-approximate RWM that converges to CE and remains efficiently samplable? A variant that achieves
$\wt{O}(T^{-1})$ convergence rate?
\end{question}

Finally, we end with a concrete direction toward answering our original question: \emph{when can one efficiently sample from RWM?} We rely in part of this work on a series of breakthroughs in the approximate sampling and high dimensional expansion literatures \cite{kaufman2020high,alev2020improved,anari2019log,anari2021entropic} leading to a sufficient condition called fractional log-concavity \cite{alimohammadi2021fractionally, anari2021entropic} for sampling hypergraphs under arbitrary external fields (generalizing an earlier result for matroids \cite{anari2019log}). This is in fact a stronger guarantee than we actually need to ensure efficient sampling for RWM. Not only are the fields we study typically additionally structured (e.g.\ monotonic), but we are also okay with some amount of decay in the mixing time depending (logarithmically) on the field size. Is there a characterization of such objects in terms of geometry of polynomials or high dimensional expansion?

\begin{question}
Is there a general condition on hypergraphs (e.g. in terms of high dimensional expansion, geometry of polynomials) that allows for approximate sampling under external fields with polylogarithmic dependence on the worst field size? What about under structural constraints (e.g.\ monotonicity)?
\end{question}

\subsection{Further Related Work}

\subsubsection{No-regret learning with structured loss}

Online learning over exponentially large classes with structured losses has been considered previously in other contexts (e.g. \cite{kleinberg2005online, cohen2017tight, helmbold2009learning,koolen2010hedging,audibert2014regret,vu2019combinatorial}). Much of this work considers the combinatorial bandit setting \cite{cesa2012combinatorial}, which typically competes against a non-adaptive adversary, but has restricted information. This is an interesting setting in its own right, but differs substantially from the challenges seen in this work and does not lead to equilibrium computation. On the other hand, there are two works which also consider efficient implementation of RWM \cite{warmuth2008randomized,helmbold2009learning}, but only for the very special settings of $k$-sets and permutations (which are generalized by our framework). Also  of note is the later work of \cite{koolen2010hedging}, who built a new hedge-based algorithm for these settings called component-hedge that also gives efficient online learning in a few additional cases (e.g.\ for spanning trees).



\subsubsection{Computing equilibra for Colonel Blotto}

The Colonel Blotto game is one of the most well studied  problems in algorithmic game theory---we restrict our attention here to some of the most notable and relevant results. As mentioned previously, the first known algorithm to compute exact Nash equilibria strategies for discrete Colonel Blotto was introduced in \cite{ahmadinejad_dehghani_hajiaghayi_lucier_mahini_seddighin_2019}, who consider games that are asymmetric across battles and across players (allowing different troop capacity and rewards across battles and players). This work remains the only known algorithm for exact equilibrium computation that is polynomial in the number of troops and battlefields, though follow-up work gave a more practical (but potentially exponential time) algorithm \cite{behnezhad_dehghani_derakhshan_hajiaghayi_seddighin_2017}.

Due to the difficulty of understanding the discrete version, a number of works have also considered Colonel Blotto's continuous relaxation. It should be noted that the equilibria in the continuous version do not apply to the discrete case \cite{perchet2022}. The works by \cite{gross1950symmetric, gross1950continuous, laslier2002two, thomas2018n, kovenock2021generalizations} consider the case that troop counts are identical (symmetric) for both players. Later, the symmetric case was also studied when Colonels have different values for battles \cite{kovenock2021generalizations}. On the other hand, when the troop counts of the two players differ, constructing/computing equilibria becomes more complicated. In \cite{roberson_2006}, the author constructs equilibrium strategies explicitly in the case that the rewards for each battlefield are the same (homogeneous). The authors of \cite{macdonell2015waging} consider Blotto with heterogeneous rewards and asymmetric troops counts (but with only two battles). In \cite{schwartz2014heterogeneous}, they consider more than two battles but with strict assumptions on the battle weights. More recently, the authors in \cite{perchet2022} present an efficient algorithm to compute approximate Nash equilibria in the two-player continuous Colonel Blotto game with asymmetric troop and battle values. 
 
There is also a breadth of work that constructs strategies for approximate and exact equilibrium under constrained parameter settings of Borel's two-player discrete version. Beginning with \cite{hart2008discrete}, the author constructs optimal strategies explicitly when the troop counts and battle rewards are identical. In \cite{vu_loiseau_silva_2018}, the authors give an algorithm to compute equilibria with fixed approximation (decaying with the number of battles). They also give an algorithm to compute the best-response strategy to a given distribution over soldiers in each battlefield using dynamic programming. In \cite{multiplayerBlotto_2020,thomas_2017}, the authors describe equilibria in the symmetric case where the number of soldiers is the same for both players. Moreover, \cite{multiplayerBlotto_2020} introduces the multiplayer variant. In \cite{vu2021colonel}, the authors construct equilbria under particular conditions for an extension of the Colonel Blotto game that accounts for pre-allocations and resource effectiveness.

\section{Preliminaries}
All throughout the paper, for integers $a \le b$ we denote by $[a,b]$ the set $\{a,\ldots, b\}$ and shorthand $[n]=[1,n]$.
We use the notations $\wt O(f)$ to hide polylogarithmic dependencies on the argument. Given a finite set $\Omega$, we denote by $\Delta(\Omega)$ the (convex) polytope of all distributions defined on $\Omega$. 
We denote by $2^{\Omega}$ as the power set of $\Omega$, i.e. the set of all subsets of $\Omega$.
Given two finite sets $\Omega_1, \Omega_2$, we denote by $\Omega_1 \times \Omega_2$ as the Cartesian product of the two sets, i.e. $(x,y) \in \Omega_1 \times \Omega_2$ if $x \in \Omega_1$ and $y \in \Omega_2$. We will use ${\Omega \choose k}$ to denote all size-$k$ subsets of the ground set $\Omega$.
Given two integers $k, n \in \mathbb{Z}^{+}$, we will use $P_k(n)$ to represent the set of ordered size-$k$ partition of $n$.
\subsection{Game Theory}\label{sec:game}
\begin{definition}[Multiplayer Simultaneous Game]
\label{def:twoplayergame}
An $m$-player \textbf{Simultaneous Game} is a tuple 
$\{\{A_i\}_{i=1}^m,\{R_i\}_{i=1}^m\}$ where $ A_i$ denotes the finite set of actions available for the $i$-th player and $R_i:  \Actions_1 \times  \cdots \times  \Actions_m \mapsto \R$ denotes the reward function for the $i$-th player. 
\end{definition}
\noindent Given a set of actions $a_1, \ldots, a_m$, we often write $a_{-i}$ to represent the combined action tuples without $a_i$, i.e. $(a_1,\ldots, a_{i-1}, a_{i+1},\ldots, a_m)$, and $R_i(a_i, a_{-i}) = R_i(a_1,\ldots, a_m)$ where we have abused notation in the input ordering to $R_i$ for simplicity of notation.

In a game, a player can choose to play an action, often called a pure strategy, or to draw randomly from a \emph{mixed strategy} given by a probability distribution over the set of available actions. 
\begin{definition}[Mixed Strategy]
\label{def:mixedstrat}
Let $\{\{A_i\}_{i=1}^m,\{R_i\}_{i=1}^m\}$ be an $m$-player simultaneous game. For the $i$-th player, the set of mixed strategies are all possible probability distributions over the actions $\Actions_i$.
Let $\vect s_i \in \Delta \lp( \Actions_i \rp)$ be the mixed strategy chosen by the $i$-th player. Then, the expected reward received by the $i$-th player is given by 
$ \E_{ a_ 1 \sim \vect s_1,\ldots, a_m \sim \vect s_m }\lp[  R_i(  a_1,\ldots, a_i,\ldots, a_m ) \rp] $.
\end{definition}
We will also make use of the following notion of a joint strategy.
\begin{definition}[Joint Strategy]
A joint strategy is a distribution $\sigma \in \Delta \lp( \Actions_1 \times \cdots \times \Actions_m \rp)$. If players were to participate in a joint strategy, then a central coordinator samples an action tuple $a = (a_1,\ldots, a_m) \sim \sigma$, and each player then plays the action $a_i$ correspondingly. As a result, the expected reward of the $i$-th player is given by $\E_{  a \sim \sigma  } R_i(a)$.
\end{definition}
For a set of actions $a_{i}^{(t)}$ for $i \in [m]$, and $t \in [T]$, we will often write $\frac{1}{T} \sum_{t=1}^T a_{i}^{(t)}$ as the mixed strategy of player $i$ such that action $a_{i}^{(t)}$ is played with probability $1/T$, and 
$\frac{1}{T} \sum_{t=1}^T a_{1}^{(t)} \otimes \cdots \otimes a_{m}^{(t)}$ as the joint mixed strategy such that the action tuple $( a_{1}^{(t)},\ldots, a_{m}^{(t)} )$ is played with probability $1/T$.

It is well known that if all players play a game optimally, the resulting strategy tuples compose of a \emph{Nash Equilibrium} of the game.
\begin{definition}[Nash Equilibrium]
In an $m$-player game $\{\{A_i\}_{i=1}^m,\{R_i\}_{i=1}^m\}$, a tuple of mixed strategies $\lp( s_1,\ldots, s_m \rp)$ composes an $\eps$-Nash Equilibria ($\varepsilon$-NE) if for all $i \in [m]$ it satisfies:
$$
\E_{s_1,\ldots, s_m}\lp[ R_i(s_1,\ldots, s_i,\ldots, s_m)  \rp]  
\geq 
\sup_{ s' \in \Delta\lp( A_i \rp)  }\E_{s_1,\ldots, s_m}\lp[ R_i(s_1,\ldots, s',\ldots, s_m)  \rp] - \eps,
$$
where the mixed strategies $s_i$  for $i \in [m]$ are mutually independent.
\end{definition}
In multiplayer and general-sum games, computating Nash equilibria is challenging. In fact, this problem is known to be complete for PPAD \cite{daskalakis_goldberg_papadimitriou_2006}, a complexity class containing many other computationally hard problems. A standard and arguably more realistic goal is to find the so-called Coarse Correlated Equilibria (CCEs) of the multi-player game, a relaxation of Nash Equilibrium introduced by Aumann \citep{aumann_1974}.

In a CCE, all players together sample from a joint mixed strategy (in contrast to NE where players \textit{independently} sample from their own mixed strategy).
Although a player $i$ cannot benefit from switching to any single action $s_i^\prime$ before the joint strategy is sampled, once a strategy $s_i$ is sampled from a CCE distribution (becoming known to each player), a player may improve their outcome by deviating (using the fact that her strategy is correlated with other players'). Thus, CCE apply to situations where a player must commit to their strategy up front and are unable to deviate after sampling.
\begin{definition}
Let $\mathcal I = \GAME$ be an $m$-player game.
An \textbf{$\eps$-approximate coarse correlated equilibrium} ($\eps$-CCE) is a joint mixed strategy $\vect \sigma \in \Delta( \Actions_1 \times \cdots \times \Actions_m )$ that satisfies:
\begin{align*}
    \forall i \in [m], \text{ and actions $a_i^\prime \in \Actions_i$}: \quad \mathbb{E}_{a \sim \vect \sigma} R_i(a) \geq \mathbb{E}_{a \sim \vect \sigma} R_i(a_i^\prime, a_{-i}) - \eps.
\end{align*}
\end{definition}

\subsection{Linear Hypergraph Game}
\label{sec:linear_hypergraph_game}
Given a ground set of vertices $\Omega$, a hypergraph $H$ is a collection of subsets of $\Omega$ called hyperedges.  
If all hyperedges of $H$ are in ${\Omega \choose k}$, the graph is called $k$-uniform.
In this work we study a special class of games whose reward functions can be `decomposed' based on the underlying structure of the game's action space. More formally, we consider games played over $k$-uniform hypergraphs whose rewards \emph{decompose linearly} over vertices. We denote this class of games as Linear Hypergraph Games.
\begin{definition}[Linear Hypergraph Games]
Let $\mathcal{I}=\GAME$ be an $m$-player game. We call $\mathcal{I}$ a \textbf{linear hypergraph game} if for all $i \in [m]$ there is a groundset $\Omega_i$ and parameter $k_i \in \mathbb{N}$ such that $A_i \subset {\Omega_i \choose k_i}$ and a `vertex-wise' reward function $R_i^{\Omega_i}:\Omega_i \times A_{-i}$ such that for all $t_1,\ldots,t_n \in A_1 \times \ldots \times A_n$
\[
R_i(t_i, t_{-i}) = \sum\limits_{v \in t_i} R^\Omega_i(v,t_{-i}).
\]
\end{definition}
\noindent In other words, each element $v$ in the ground set $\Omega_i$ has a certain reward with respect to any choice of the opponents, and the reward of a $k$-set is simply the sum of its individual rewards. Many important games that are well-studied in the game theory literature falls under this category, e.g.\ Colonel Blotto Games, Security Games, Congestion Games, Dueling Games, etc. 
In fact, it should be noted a similar notion has been studied in the online learning setting in \cite{koolen2010hedging}, who develop an efficient no-regret algorithm called Component Hedge for linear losses over basic structures such as the complete hypergraph, truncated permutations, and spanning trees.

\subsection{No-Regret Learning in Games}
We consider the framework of No-Regret Learning in Games (see \cite{cesa2006prediction,daskalakis2021near} and references therein). 
In this framework, a game is iterated with one or more players implementing a no-regret learning algorithm to adaptively choose strategies. 
At the $t$-th round of the game, each player selects a mixed strategy $\vect s_{i}^{(t)}$, 
and samples the action $a_i^{(t)} \sim \vect s_{i}^{(t)}$, where the choice of $\vect s_{i}^{(t)}$ depends only on $a_j^{(t')}$ for $j \in [m]$ and $t' < t$.

The goal for each player is to optimize her regret, defined as the following.
\begin{definition}[Regret]
At the $T$-th round of the game, the regret for the $i$-th player is defined as
$$
\Reg_{T,i}
\eqdef 
\max_{ a^* \in \Actions_i }
\sum_{t=1}^T R_i( a_1^{(t)},\cdots, a_{i-1}^{(t)}, a^*, a_{i+1}^{(t)},\cdots, a_m^{(t)}  )
-
\sum_{t=1}^T R_i( a_1^{(t)}, \cdots, a_m^{(t)}  ).
$$
\end{definition}
It is classical result that if all players follow no-regret learning strategies, the overall dynamics quickly converge to a Nash or Coarse-Correlated
Equilibria (CCEs) of the game (see e.g.\ \cite{freund_schapire_1999,cesa2006prediction,ostrovski2013payoff}).
\begin{theorem}[No-Regret Implies Equilibrium Computation \cite{freund_schapire_1999}]
\label{thm:no-regret-cce}
Suppose $m$ players are playing under the No-Regret Learning in Games framework.
Let $\vect \sigma^* \eqdef \frac{1}{T} \sum_{t=1}^T a_1^{(t)} \otimes \cdots \otimes a_m^{(t)}  $
be the average mixed joint strategies played by the $i$-th player over $T$ rounds. 
Then, $\vect \sigma^*$ forms an 
$T^{-1} \max\left( \Reg_{T,1},\cdots,\Reg_{T,m}\right)$-approximate CCE of the game, where
$R_i^{(T)}$ is the regret for the $i$-th player at the $T$-th round.
When $m=2$ and the game is zero-sum, the mixed strategies 
$\left(\frac{1}{T} \sum_{t=1}^T a_1^{(t)}, \frac{1}{T} \sum_{t=1}^T a_2^{(t)}\right)$ constitute a $T^{-1} \max \lp(  \Reg_{T,1} ,\Reg_{T,2} \rp)$-approximate Nash Equilibrium.
\end{theorem}
\subsection{Randomized Weighted Majority Algorithm}

As in online learning, no-regret learning in games studies the regret of a player against the opponents' strategies in repeated play with respect to the best single strategy in hindsight. 

One of the most frequently used tools in no-regret learning is the randomized weighted majority (RWM) algorithm. For player $i \in [m]$, RWM maintains the mixed strategies from $\Delta\lp( \Actions_i \rp)$ as follows: at the first round, it chooses uniformly among the actions. 

At the $(T+1)$-st round, a cumulative reward is computed for each action $x \in \Actions_i$
$$
r^{(T+1)}(x) = \sum_{t=1}^T R_i(s_1^{(t)}, \cdots, s_{i-1}^{(t)}, x, s_{i+1}^{(t)}, \cdots, s_m^{(t)}),
$$
and RWM chooses the mixed strategy $\HD{T+1}{\beta}$ (which we refer to as the RWM distribution) such that 
\begin{align}
\label{eq:OH-def}
\Pr\lp[ \HD{T+1}{\beta} = x \rp] \propto \beta^{-r^{(T+1)}(x)}.
\end{align}

It is well known that if any player samples according to $\HD{T+1}{\beta}$ in each round, her expected regret will be bounded in the worst case by $O_{\beta,N}(\sqrt{T})$.

In games whose action sets are exponentially large, exactly sampling from the the RWM distribution may be intractable in relevant cases. Nonetheless, we show that similar regret bounds hold even when one \emph{approximately} samples the RWM distributions in each round (the proof is given in Appendix). 
\begin{definition}[Approximate Sampling]
We say a randomized algorithm $\mathcal{A}$ with output space $\Omega$ $\delta$-approximately samples a distribution $\mu$ over $\Omega$ if the output of $\mathcal{A}$ is $\delta$-close to $\mu$ in TV-distance.
\end{definition}
We call any strategy that $\delta$-approximately samples from $\HD{t}{\beta}$ in each of $T$ rounds of repeated play \textit{$\delta$-approximate RWM}, and denote this class of algorithms by $\delta$-RWM$^T_\beta$. It is not hard to show that $\delta$-RWM has near-optimal regret in the adversarial setting (see \Cref{app:nr}).
\begin{lemma}[$\delta$-RWM is No-Regret]
\label{lem:apx-mwu} 
Let $\mathcal{I}$ be an $m$-player game with at most $N$ actions and $\Lmax$ reward. If the $i$-th player follows $\delta$-RWM$^T_\beta$ in $T$ rounds of play with learning rate $\beta =
1 - \sqrt{ \log(N) / T}$ and approximation factor
$\delta \leq \sqrt{  \log (N) / T }$,\footnote{We are assuming $T \gg \log N$. Otherwise, the regret bound becomes $\Lmax T$, which can be achieved by any arbitrary sequence of choices. }
then for any $\eta>0$ they experience regret at most 
$$
\Reg_{T,i} \leq O\lp( \Lmax \sqrt{T} \lp( \sqrt{ \log N} + \sqrt{ \log(1/\eta) }  \rp)\rp)
$$ 
with probability at least $1 - \eta$.
\end{lemma}
\noindent As an immediate corollary, 
for any game, if we can approximate sample from the RWM distribution efficiently, we immediately get an efficient no-regret learner.
In addition, connecting it with \Cref{thm:no-regret-cce}, we also obtain the following corollary for equilibrium computation with RWM.
\begin{corollary}[Equilibrium Computation with $\delta$-RWM]
\label{cor:compute-cce}
Let $\mathcal I = \GAME$ be an $m$-player where $|\Actions_i| \leq N$ for each $i \in [m]$ and reward bounded by $\Lmax$.
Suppose the game is played repeatedly for $T$ rounds.
If there is an algorithm which can perform $\delta$-RWM$^T_\beta$ for each player $i$ in time
$f_{\mathcal I}( T, \delta )$, then there exists an algorithm which computes an $\eps$-CCE of $\mathcal I$ (Nash if the game is $2$-player zero-sum) with probability at least $1 - \eta$. Moreover, the algorithm runs in time
$$
O \lp( 
m \cdot f_{\mathcal I}\lp(  C \cdot \Lmax \eps^{-2}  \log (Nm/\eta), \frac{\eps}{C \cdot \Lmax} \rp)
\rp).
$$
for some universal constant $C>0$.
\end{corollary}
In Sections \ref{sec:mcmc-sample} and \ref{sec:dp-sample}, we develop two different types of methods of approximate sampling from the RWM distribution in many well-studied games and discuss their implications. Before moving on, however, it is convenient to briefly discuss one computational consideration that frequently occurs in efficient implementation of RWM. In particular, it will often be the case that our algorithm needs to deal with piece-wise constant functions that map from $[0,n]$ to $\R$ (e.g. reward functions in Blotto for each battlefield are $2$-piecewise in this sense). To represent such functions, we will use the following data structure that we refer as a \emph{succinct representation}.
\begin{definition}[Succinct Descriptions of piecewise constant functions]
Let $f: [0,n] \mapsto \R $ be a $q$-piecewise constant function.  
The succinct description of $f$, denoted as $D_{f}$, consists of $q$ tuples of the form $(a_i,b_i,y_i) \in (\Z^+, \Z^+, \R) $ such that for all $x \in [a_i, b_i]$, $f(x) = y_i$ and the intervals $\{[a_1, b_1] \cdots [a_q, b_q]\}$ partition $[0, n]$.

\end{definition}
We will often write $|D_f|$ to denote the number of intervals contained in the succinct description. Finally, note that assuming access to succinct descriptions does not lose much generality, as given query access to a standard representation for the monotonic piece-wise function in question (e.g.\ in the RAM model), it is easy to build a succinct description in time $q\log(n)$ by binary search. 

\section{Playing Games via MCMC-Sampling}\label{sec:mcmc-sample}
In this section, we develop the connections between linear hypergraph games, the RWM distribution, and efficient sampling. In doing so, we unlock access to powerful tools from the sampling literature for the first time in the context of games. This allows for a number of immediate applications including the first no-regret algorithms for well-studied settings such as matroids.

With this in mind, let's first recall the basic framework of Monte Carlo Markov Chains: a powerful tool for approximately sampling from large spaces like the RWM distribution. More formally, consider the following problem: given a distribution $\pi$ over a large state space $\Actions$, we'd like to \textit{approximately} sample a state from $\pi$ in $\text{polylog}(|\Actions|)$ time. MCMC-sampling is an elegant approach to this problem in which one defines a Markov chain $M$ on $\Actions$ satisfying the following three conditions:
\begin{enumerate}
    \item The stationary distribution of $M$ is $\pi$
    \item A single step of $M$ can be implemented efficiently
    \item $M$ converges quickly to its stationary distribution.
\end{enumerate}
As long as these three conditions hold, it is possible to efficiently sample from $\pi$ up to any desired accuracy simply by running the Markov chain from any starting position a few steps and outputting the resulting state. More formally, recall that the \textit{mixing time} of a Markov chain $M$ is the number of steps until the resulting distribution is close in TV-distance to $\pi$:
\begin{definition}[Mixing Time]
The mixing time of a Markov chain $M$ is the worst-case number of steps until the total variation distance of $M$ is close to its stationary measure:
\[
T(M,\delta) \coloneqq \min_{t \in \mathbb{N}} : \forall \pi_s, \ TV(M^t\pi_s, \pi) \leq \delta.
\]
\end{definition}
\noindent Thus one only needs to run the chain $T(M,\delta)$ times (from any starting position) in order to $\delta$-approximately sample from $\pi$.

While MCMC-sampling is a promising approach, designing efficient Markov chains is typically a challenging task. However, in structured settings such as linear hypergraph games, the distributions arising from RWM seem to be more conducive to the approach. In fact, they correspond to well-studied structure in the approximate sampling literature called \textit{external fields}.
\begin{definition}[External Field]
Let $\pi$ be a distribution over a $k$-uniform hypergraph $H \subset {\Omega \choose k}$. The distribution given by $\pi$ `under external field $w$' for $w \in \R_+^\Omega$ has measure proportional to the product of w across each k-set:
\[
\pi^w(s) \ \propto \ \pi(s) \prod\limits_{v \in s} w(v).
\]
When $\pi$ is uniform over $H$, we often just write $H^w$ for $\pi^w$.
\end{definition}
It is a simple observation that the distribution arising from RWM (or indeed any reasonable variant) is exactly given by the application of an external field to the state space. 
\begin{observation}[RWM $\to$ External Fields]\label{obs:hedge-to-external}
Let $\mathcal{I}=\GAME$ be an $m$-player linear hypergraph game. Then for any $i \in [m]$, $\HD{T}{\beta}$ can be written as the application of an external field $w$ to $A_i$ such that $A_i^w$ has minimum probability at most $\frac{\beta^{-2\Lmax T}}{|A_i|}$.
\end{observation}
\begin{proof}
Assume $i=1$ without loss of generality (for simplicity of notation). Recall that RWM operates at round $T+1$ by exponentiating the total loss over the previous $T$ rounds:
\[
\Pr \lp[ \HD{T}{\beta} = a \rp] \propto  \beta^{-\sum\limits_{j=1}^T R_1(a,s^{(j)})}
\]
where $s^{(1)}, \ldots, s^{(T)} \in A_{-1}$ are the historical strategies played by players $\{2,\ldots,n\}$ in rounds one through $T$. 
For simplicity of notation, 
let $\ell^{(T)}(a) \coloneqq -\sum \limits_{t=1}^T R(a,s^{(t)})$ be the total loss. We can similarly define this quantity for any element of the ground set $v \in \Omega$ as:
\[
\ell^{(T)}(v) =  -\sum\limits_{t=1}^T R_1^\Omega(v,s^{(t)}).
\]
Switching the summations, linearity promises we can express $\ell_T(s)$ as a sum over $\ell_T(v)$:
\[
\ell_T(a) = \sum\limits_{v \in a} \ell_T(v).
\]
As a result, the RWM distribution is proportional to the product of (exponentiated) total loss for each vertex:
\[
\Pr(a) \ \propto \ \prod\limits_{v \in a} \beta^{\ell^{(T)}(v)}.
\]
This is exactly $A_1$ under the external field $w \in \R_{+}^\Omega$ where $w(v) = \beta^{\ell^{(T)}(v)}$. Since the distribution started uniformly over $A_1$ and is update by at most $\beta^{\Lmax}$ in each step, the minimum probability is at worst $\frac{\beta^{-2\Lmax}}{|A_1|}$. 
\end{proof}
As an immediate corollary, we get an efficient no-regret algorithm for any linear hypergraph game whose state space can be approximately sampled under arbitrary external fields, a well-studied problem in approximate sampling. In the remainder of the section, we'll show how such results lead to new efficient no-regret learning algorithms for many well-studied settings in game theory.
\subsection{Glauber Dynamics and Fractionally Log-Concave Games}\label{sec:matroid}
The past few years have seen major advances in approximate sampling various combinatorial objects under external fields \cite{anari2019log,anari2020spectral,alimohammadi2021fractionally,anari2021entropic}. The recent breakthroughs have largely been driven by new analysis techniques for a simple local Markov chain arising from the study of Ising models in statistical physics called the \textit{Glauber Dynamics}. Starting from a state $\sigma \in H \subset {\Omega \choose k}$, the (single-site) Glauber Dynamics for a distribution $\pi$ over $A_i$ are given by the following two-stage procedure:
\begin{enumerate}
    \item ``Down-Step:'' Remove a vertex $v$ uniformly at random from $\sigma$.
    \item ``Up-Step:'' Sample from $\pi$ conditional on $\sigma \setminus v$.
\end{enumerate}
It is not hard to show that $\pi$ is the stationary distribution of this process. The first major breakthrough towards rapid mixing of Glauber Dynamics was due to Anari, Liu, Oveis-Gharan, and Vinzant \cite{anari2019log}, who used tools developed in the high dimensional expansion literature \cite{kaufman2020high} to prove rapid-mixing of Glauber Dynamics on a broad class of combinatorial objects called \textit{matroids}. 

\begin{definition}[Matroids]
Let $\Omega$ be a ground set and $\mathcal J$ a family of subsets of $\Omega$. $(\Omega,\mathcal J)$ is called a matroid if it satisfies
\begin{enumerate}
    \item \textbf{Non-emptiness:} $\mathcal J$ contains at least one subset.
    \item \textbf{Downward-closure:} For all $S \in \mathcal J$ and $S' \subset S$, $S' \in \mathcal J$
    \item \textbf{Exchange-property:} For all $S,S' \in \mathcal J$ s.t. $|S|>|S'|$, there exists $x \in S$ s.t. $S' \cup x \in \mathcal J$.
\end{enumerate}
An element $S \in I$ is a \textbf{basis} if it is maximal, and the \textbf{rank} of the matroid (denoted $r(\mathcal J)$) is the size of its largest basis.
\end{definition}
Note that the bases of a matroid make up an $r(\mathcal J)$-uniform hypergraph over vertex set $\Omega$. These objects are perhaps best thought of as generalizing the combinatorial structure seen in spanning trees (which form the bases of a `graphic' matroid). ALOV's \cite{anari2019log} major breakthrough was to prove rapid mixing of Glauber Dynamics on matroid bases, a problem known as the Mihail-Vazirani Conjecture (this result was later optimized by Cryan, Guo, and Mousa \cite{cryan2019modified}).
Since matroids maintain their structure under external fields (see e.g.\ \cite{anari2021entropic}), this leads to the following MCMC-algorithm for sampling matroid bases under arbitrary external fields.
\begin{theorem}[Glauber Dynamics on Matroids {\cite[Theorem 5]{anari2021entropic}}]\label{thm:FLC-glauber}
Let $H$ be the set of bases of a rank-$k$ matroid $(\Omega, \mathcal J)$.
Let $w \in \mathbb{R}_{+}^{\Omega}$ be any external field.
Then, the single-step Glauber Dynamics on $H^w$ has mixing time
\[
T(GD,\delta) \leq O\left(
k\log\left(\frac{\log(|\Omega|/w_*)}{\delta}\right)\right).
\]
\end{theorem}
A substantial amount of progress has been made since ALOV and CGM's works. In fact, recently Anari, Jain, Koehler, Pham, and Vuong \cite{anari2021entropic} introduced an even more general class of hypergraphs that can be sampled under arbitrary external fields called fractionally log-concave hypergraphs.\footnote{Formally this requires a slight generalization known as the $q$-step Glauber Dynamics.} All of our results extend to this setting, but to our knowledge matroids already capture most settings of interest in game theory so we focus just on this case for concreteness.

Since matroids are typically exponential size in their rank, we will need implicit access in order to build efficient algorithms. This is typically done through various types of oracle access to the independent sets. For simplicity of presentation, we will assume access to a \textit{contraction oracle}, a standard operation on matroids that restricts the object to independent sets containing some fixed $S \in \mathcal{J}$.
\begin{definition}[Contraction Oracle]
Let $(\Omega,\mathcal{J})$ be a rank-$k$ matroid. A rank-$r$ contraction oracle inputs an independent set $S \in \mathcal{J}$ of size $r$, and outputs (query access to) the contracted matroid $(\Omega,\mathcal{J}_S)$, where
\[
\mathcal{J}_S = \{ T : T \cup S \in \mathcal{J}\}.
\]
\end{definition}
Crucially we will only use rank-$(k-1)$ contraction oracles on matroids with rank $k$ (and thus drop rank from the notation below). This can always be implemented in $|\Omega|$ applications of a standard independence oracle (which decides given $S \subseteq \Omega$ whether $S \in \mathcal{J}$), but can often be implemented much more efficiently. For instance, it is easy to see in the case of uniform (unconstrainted) matroids, this can be implemented in $O(k)$ time simply by removing each element of $S$ from the list.

Before stating the main guarantees for no-regret learning and equilibrium computation, it will be useful to discuss a few finer-grained properties typical to games on matroids that help parameterize related computation complexities. First, while not strictly necessary, it will be convenient to restrict our attention to games where the action sets of all players are given by $k_i$-uniform hypergraphs on some shared groundset $\Omega$, i.e.\ $\Actions_i \subseteq {\Omega \choose k_i}$ for all $i$. Given such a game $\mathcal{I} = \GAME$, we will typically write $k_{-i} = \max_{j \neq i} k_j$ to denote the maximum support of any viable opponent strategy.

Second, it will be useful to introduce an important property of congestion and security games we call \textit{collision-sensitivity}: the vertex-wise reward of an element $v \in \Omega$ only changes if $v$ is also selected by another opponent.
\begin{definition}[Collision-sensitive Games]
\label{def:collision-sensitive}
Let $\mathcal{I}=\GAME$ be an $m$-player linear hypergraph game where $\Actions_i \subseteq 2^{\Omega}$. We call the rewards of player $i$ `collision-sensitive' if for all $v \in \Omega$, the vertex-wise reward of $v$ only changes if another opponent also selects $v$:
    \[
    \forall v \in \Omega \text{ and } s,s' \in A_{-i} \text{ s.t. } v \notin s,s': R_i^\Omega(v,s)=R_i^\Omega(v,s').
    \]
We will write $\NC_i: \Omega \mapsto \R$ as the function specifying the $i$-th player's no-collision reward function for each vertex, i.e. $\NC_i(x) = R_i^{\Omega}(x, s)$ for $x \notin s$. We say a collision-sensitive reward has support $q$ if the no-collision reward function $\NC_i$ for each player takes on at most $q$ values across all vertices $v \in \Omega$.
\end{definition}
In a sense, collision-sensitivity can be thought of as an independence criterion on the vertices: roughly speaking, actions taken on $v$ do not effect actions taken on $w$ for $w\neq v$. With these definitions out of the way, we can now state our main guarantees for no-regret learning on matroids.
\begin{theorem}[RWM on Matroids]\label{thm:matroid-no-regret}
Let $\mathcal{I}=\GAME$ be an $m$-player linear hypergraph game on a size-$n$ ground set $\Omega$. If $A_i$ is collision-sensitive with support $q$, then it is possible to implement $\delta$-RWM$^T_\beta$ in time
\[
O\left(k_iT(CO + q\log(n)+mk_{-i}T\log(n)) \log\left( \frac{k_i\log(n) + \Lmax T\log(\beta^{-1})}{\delta}\right)\right),
\]
assuming access to a $q$-piecewise succinct description of $\NC_i$ encoded under an ordering of $\Omega$ and a contraction oracle matching the same ordering.
\end{theorem}
The proof of \Cref{thm:matroid-no-regret} is not particularly interesting beyond combining \Cref{obs:hedge-to-external} and \Cref{thm:FLC-glauber} and involves mostly  tedious implementation details of Glauber Dynamics on matroids. We give these details in \Cref{app:Glauber-implementation} for completeness.

Before moving on, we briefly note this result is nearly tight in many of the main parameters. For instance, the dependence on $k$,$m$, and $n$ is $\tilde{O}(mk\log(n))$,\footnote{This may increase when the support of other players is non-constant. E.g.\ if all players are playing bases of $k$-uniform matroids, we require $mk^2\log(n)$} which in many cases (e.g. uniform matroid) is the number of bits required even to express a set of pure strategies for each player. The bound is also linear in $q$, which is easily seen to be necessary since one needs to know the $q$ distinct values in order to sample.

Many games in the literature satisfy the conditions of \Cref{thm:matroid-no-regret}. We'll end this subsection by giving a few concrete examples. Perhaps the most well-studied variant of these games is a popular setting called \textit{congestion games}. Congestion games are a natural model for resource competition where $m$-players compete to share $n$ resources and receive rewards dependent on the number of players sharing the same resource.
\begin{definition}[Congestion Game]
Given a ground set $\Omega$, an $m$-player congestion game on $\Omega$ consists of a collection $\{A_i\}_{i=1}^m$ and a reward function $c: \Omega \times [m] \to \mathbb{R}$ where each $A_i \subseteq 2^{\Omega}$ are the strategies of player $i$, and the reward on the actions $(s_1,\ldots,s_n)$ is given by:
\[
R_i(s_i, s_{-i}) = \sum\limits_{e \in s_i} c(e,|e(s)|)
\]
where $e(s) = \{s_i: e \in s_i\}$.
\end{definition}
Congestion games are particularly well-studied on matroid bases, which are the only structure on which best response is known to converge to Nash in polynomial time. However, to our knowledge \Cref{thm:matroid-no-regret} provides the first no-regret algorithm for congestion games.
\begin{corollary}[Matroid Congestion without Regret]\label{cor:congestion-nr}
Let $\mathcal{I}=\{ \{\Actions_i\}_{i=1}^m,c\}$ be a congestion game 
where each $\Actions_i$ is the set of bases of a rank-$k_i$ matroid on a ground set $\Omega$ of size $n$ satisfying $\max_{i \in [m]} k_i = k$.
Suppose $\NC(e) = c(e,1)$, the no collision reward function, is $q$-piecewise under some ordering of $\Omega$. 
Then there is a no-regret learning algorithm for $\mathcal{I}$ with regret:
\[
\Reg_T \leq O\left(\Lmax \sqrt{T} \cdot \lp( \sqrt{k\log(n)} + \sqrt{\log(1 / \eta)} \rp) \right)
\]
with probability at least $1-\eta$ that runs in time \[O\left(kT(CO + q\log(n)+Tmk\log(n) ) \log\left(\Lmax T k\log(n)\right)\right)\] 
assuming access to a $q$-piecewise succinct description of $\NC$ encoded in some ordering of $\Omega$ and a contraction oracle matching this ordering.
\end{corollary}
\begin{proof}
It is enough to argue the game is linear, as the result then follows immediately from \Cref{thm:matroid-no-regret}. 
Denote by $s$ be the strategy tuples chosen by the players.
Recall that the reward of any strategy $s_i \in \Actions_i \subset {\Omega \choose k_i}$ in the congestion game is given by:
\[
R(s_i, s_{-i}) = \sum\limits_{e \in s_i} c(e,|e(s)|).
\]
$c$ can easily be extended into the desired vertex-wise reward function, so we are done.
\end{proof}
\begin{corollary}[Equilibrium Computation for Matroid Congestion]\label{cor:congestion-ec}
Let $\mathcal{I}=\{ \{\Actions_i\}_{i=1}^m,c\}$ be a congestion game 
where each $\Actions_i$ is the set of bases of a rank-$k_i$ matroid on a shared ground set $\Omega$ of size $n$ satisfying $\max_{i \in [m]} k_i = k$.
Suppose $\NC(e) = c(e,1)$ is $q$-piecewise under some ordering of $\Omega$. 
Then it is possible to compute an $\varepsilon$-CCE with probability at least $1-\eta$ in time
\[
O\left(m\Lmax^2k^2\log(mn/\eta)\varepsilon^{-2}(CO + q\log(n)+m\Lmax^2k^2\log^2(mn/\eta)\varepsilon^{-2}) \log\left(\Lmax k\log(mn/\eta)/\varepsilon \right)\right).
\]
\end{corollary}

We note that these results also easily generalizes matroid congestion over any FLC, unlike the best response strategy for computing Nash. Furthermore, we note that Hedge is actually known to converge to better equilibria \cite{kleinberg2009multiplicative} than original techniques based on best response, which gives this approach an additional potential advantage.

Another setting particularly well-suited to matroids are \textit{security games}, which model a variety of attack/defense scenarios.
\begin{definition}[Security Game]
A security game $\mathcal{I}=(A_d,A_a,\{r,\zeta,c,\rho\})$ over ground set $\Omega$ consists of defender actions $A_d \subseteq 2^\Omega$, attacker actions $A_a=\Omega$, and reward/cost functions $r,\zeta,c,\rho : \Omega \to \R$. Let $S \in A_d$, $i \in A_a$ be the actions taken by the defender and the attacker respectively.
The reward matrices are given by:
\[
R_d(S,i) = \begin{cases} r(i) & \text{if } i \in S\\
c(i) & \text{else,}
\end{cases} \ \ \ \ \text{and} \ \ \ \ R_a(S,i) = \begin{cases} \zeta(i) & \text{if } i \in S\\
\rho(i) & \text{else.}
\end{cases}
\]
\end{definition}
Security games can model a couple natural settings dependent on the choice of parameters. One basic setting is where the defender has $k$ security resources to defend a set of $n$ targets, and `wins' if the attacker chooses a defended target. On the other hand, the model also captures the complement of this game where the defender chooses $k$ targets to distribute key resources, and the attacker wins if they intercept this distribution (pick one of the $k$ marked targets). Security games have broad applicability in practice, and indeed have been used in cases such as assigning security checkpoints at LAX \cite{tambe2011security}.

Security games are inherently linear in their natural representation and thus admit efficient no-regret algorithms when the defender's state-space is a matroid (simulating RWM for the attacker is trivial as it corresponds to a size-$|\Omega|$ multinomial distribution). 
\begin{corollary}[Security without Regret]\label{cor:sec-nr}
Let $\mathcal{I}=(A_d,A_a,\{r,\zeta,c,\rho\})$ be a security game where $A_d$ are the bases of a rank-$k$ matroid on the ground set $\Omega$ and $A_a = \Omega$. 
Suppose $c,\rho:\Omega \mapsto \R$ are $q$-piecewise under some ordering of $\Omega$. 
Then there exists a no-regret learning algorithm for $\mathcal{I}$ with regret:
\[
\Reg_T \leq O\left(\Lmax \sqrt{T} \cdot \lp( \sqrt{k\log(n)} + \sqrt{\log(1 / \eta)} \rp) \right)
\]
with probability at least $1-\eta$ that runs in time
\[
O\left(kT(CO + q\log(n)+T\log(n)) \log\left(k\log(n)T\Lmax \right)\right),
\]
assuming access to $q$-piecewise succinct descriptions for $c$ and $\rho$ encoded in some ordering of $\Omega$ and a contraction oracle matching this ordering.
\end{corollary}
\begin{proof}
Note that the attacker's strategy consists of bases of just rank-$1$ matroid so implementing RWM for her is trivial. We focus on the implementation for the defender side.
Again it is enough to show the game is linear. By definition, we have
\[
R_D(S,i) = \sum\limits_{j \in S} R_D^\Omega(j,i)
\]
where $R_D^\Omega: [n] \times [n] \to \R$ is 
\[
R_D^\Omega(i,j) = \begin{cases}
r(i) & \text{if } i=j\\
c(i) & \text{else}.
\end{cases}
\]
\end{proof}
As an immediate corollary, we also get fast equilibrium computation.
\begin{corollary}[Equilibrium Computation for Security]\label{cor:sec-ec}
Let $\mathcal{I}=(A_d,A_a,\{r,\zeta,c,\rho\})$ be a security game where $A_d$ are the bases of a rank-$k$ matroid on the ground set $\Omega$ and $A_a = \Omega$. 
Suppose $c,\rho:\Omega \mapsto \R$ are $q$-piecewise under some ordering of $\Omega$. 
Then it is possible to compute an $\varepsilon$-CCE (Nash if the game is zero-sum) with probability at least $1-\eta$ in time
\[
O\left(\Lmax^2k^2\log(n/\eta)\varepsilon^{-2}(CO + q\log(n)+\Lmax^2k\varepsilon^{-2}\log^2(n/\eta)) \log\left(k\log(n/\eta) \eps^{-1} \Lmax \right)\right),
\]
assuming access to an $q$-piecewise succinct description for $c$ and $\rho$ encoded under some ordering of $\Omega$ and a contraction oracle matching this ordering.
\end{corollary}
We note that this result easily generalizes to settings with multiple attackers or an attacker who chooses targets corresponding to a matroid basis.
\subsection{Dueling Games and the JSV Chain}
Matroids (or more generally FLC's) are not the only type of constrained state space that can be sampled under arbitrary external fields. Indeed, long before these results Jerrum, Sinclair, and Vigoda \cite{jerrum2004polynomial} famously proved (in work on approximating the permanent) that bipartite matchings have this property as well. We give an improved version of their result due to Bez\'{a}kov\'{a}, \v{S}tefankovi\v{c}, Vazirani, and Vigoda \cite{bezakova2004approximating}.
\begin{theorem}[JSV Chain]\label{thm:jsv}
Let $(K_{n , n},w)$ be an edge-weighted complete bipartite graph, and consider the distribution over perfect matchings given by:
\[
Pr(M) \propto \prod\limits_{e \in M} w_e.
\]
It is possible to $\delta$-approximately sample from this distribution in $\wt{O}(n^7\log\frac{1}{\delta w_{\min}})$ time, where $w_{\min}$ is the minimum weight.
\end{theorem}
Note one can phrase this result as a sampling algorithm for permutations over external fields, where the state-space is viewed as a subset of $[n]^n$. Like matroids, bipartite matchings are very natural objects and underlie a fair number of well-studied games. In this section, we focus on the setting of \textit{dueling games}. Dueling games model two player competitive optimization over a shared ground set.
\begin{definition}[Dueling Games]
A dueling game $\mathcal{I}=(\Omega, \mu, A_1,A_2)$ consists of a set $\Omega$, a distribution $\mu$ over $\Omega$, and strategy spaces $A_1,A_2 \subset \R_+^{\Omega}$. The reward matrices are given by the probability of ranking $x \sim \mu$ higher than the opponent:
\[
R_1(s,t) = Pr_{x \sim \mu}[s(x) > t(x)] - Pr_{x \sim \mu}[t(x) > s(x)],
\]
and likewise:
\[
R_2(s,t) = Pr_{x \sim \mu}[t(x) > s(x)] - Pr_{x \sim \mu}[s(x) > t(x)].
\]
\end{definition}
There is no known polynomial time algorithm for computing equilibria of general dueling games. We will give a general algorithm for a class of dueling games we call \textit{unrestricted}.
\begin{definition}[Unrestricted Dueling Games]
A dueling game $\mathcal{I}=(\Omega, \mu, A_1,A_2)$ is called un-restricted if there exist subsets $S_1,S_2 \subset \mathbb{R}$ with $ |S_1| = |S_2| = |\Omega|$ such that $A_1$ (respectively $A_2$) consists of all possible assignments of $\Omega$ to $S_1$ (respectively $S_2$).
\end{definition}
It is not hard to see that unrestricted dueling games are linear over perfect matchings in a complete bipartite graph. As a result, we can use the JSV-chain to simulate optimistic hedge in polynomial time.
\begin{theorem}[Sampling Unrestricted Dueling Games]\label{thm:duel}
Let $\mathcal{I}=(\Omega, \mu, A_1,A_2)$ be an unrestricted dueling game where $|\Omega|=n$. Then it is possible to implement $\delta$-RMW$^T_\beta$ in time $\wt{O}(T^2n^7\log(1/\delta))$.
\end{theorem}
\begin{proof}
We focus on player $1$. The result is analogous for player $2$. Strategies in an unrestricted dueling game correspond to perfect matchings in the complete bipartite graph $K_{n,n}$, where the LHS corresponds to elements of $\Omega$, and the RHS corresponds to elements in $S_1$. To fit into our prior framework of linearity and external fields, one may view these perfect matchings as elements of $E^n$ (where $E$ is the edge set of $K_{n,n}$). Recall that the reward is given by:
\[
R_1(s,t) = Pr_{x \sim \mu}[s(x) > t(x)] - Pr_{x \sim \mu}[t(x) > s(x)]
\]
It is not hard to see this is linear over the edges of matching:
\[
R_1(s,t) = \sum\limits_{e \in s}R_1^\Omega(e,t)
\]
where $R_1^\Omega: E \times A_2 \to \R$ is given by:
\[
R_1^\Omega(\{v,w\},t) = \begin{cases}
\mu(v) & \text{if } w > t(v)\\
-\mu(v) & \text{if } w < t(v)\\
0 & \text{else}
\end{cases}
\]
As a result, \Cref{obs:hedge-to-external} implies that RWM is given by the application of an external field over the edges of perfect bipartite matchings, which is exactly the distribution considered in \Cref{thm:jsv}. All that is left is to efficiently build access to the weights of the underlying bipartite graph, which is a small onetime cost that is asymptotically dominated by even the mixing time of the JSV chain. As a result, it is enough to run \Cref{thm:jsv} $T$ times, which gives the resulting runtime bound.
\end{proof}
We note that this result can easily be generalized to a slightly larger class of games where $|S_1|$ and $|S_2|$ may be larger than $\Omega$, and specific edges in the bipartite representation may be disallowed (i.e. we might add the constraint that $x \in \Omega$ can never be given rank $1$). Such strategies correspond to sampling matchings on a generic bipartite graph (rather than $K_{n,n}$), and no-regret learning can can also be performed by the JSV-chain.

Finally, we'll look at a classic dueling game that fit into the unrestricted framework: ranking duel. Ranking duel (or the `search engine game') is a game where two players compete to choose the best ranking of $n$ items. One of these items is pulled from a known distribution, and the player who ranked it higher wins.

Ranking duel is an unrestricted game where $S_1,S_2=[n]$ and the action spaces $\Actions_1, \Actions_2 = \mathcal S_n$, i.e. permutations of $n$. As a result \Cref{thm:duel} immediately implies an efficient algorithm for sampling in $\delta$-RWM.
\begin{corollary}[Ranking Duel without Regret]\label{cor:duel-nr}
Let $\mathcal{I}=([n],\mu,\mathcal S_n,\mathcal S_n)$ be an instance of ranking duel. Then there exists an algorithm with regret:
\[
\Reg_T \leq O\left(\sqrt{T} \cdot \lp( \sqrt{k\log(n)} + \sqrt{\log(1 / \eta)} \rp) \right)
\]
with probability at least $1-\eta$ that runs in time $\wt{O}(T^2n^7)$.
\end{corollary}
As a corollary we get the fastest known equilibrium computation for ranking duel,
\begin{corollary}[Equilibrium Computation for Ranking Duel]\label{cor:duel-ec}
Let $\mathcal{I}=([n],\mu,S_n,S_n)$ be an instance of ranking duel. Then there exists an algorithm computing an $\varepsilon$-CCE (Nash if the game is zero-sum) with probability at least $1-\eta$ in time  $\tilde{O}(n^9\log(1/\eta)/\varepsilon^4)$
\end{corollary}

Unfortunately, while the JSV-chain is an improvement over previous extended linear programming approaches to dueling games \cite{immorlica2011dueling,ahmadinejad_dehghani_hajiaghayi_lucier_mahini_seddighin_2019}, $n^9$ can hardly be called a practical running time. In fact, it should be noted there is a faster known no-regret algorithm for perfect bipartite matchings called \textsc{PermELearn} that runs in $O(Tn^4)$ time. 

Thus \Cref{thm:duel} is perhaps more interesting from the perspective of the method than the result itself. Designing faster and simpler Markov chains for sampling bipartite matchings has long been a favorite open problem in the sampling community. Our setting gives a nice intermediate version of this problem, as the matching problems arising from unrestricted dueling games have particularly nice structure. In particular, they correspond to \textit{monotonic weightings}, in the sense that for every fixed vertex $v$ on the LHS on the graph, the edge weight of $w(v,i) \geq w(v,j)$ if $i \geq j$. Matchings with monotonic weights are actually well-studied in the literature, including the resolution of the monotone column permanent conjecture \cite{branden2011proof} and rapid mixing of the switch chain\footnote{The switch chain is in essence the $2$-step Glauber Dynamics on the view of matchings as a subset of $E^n$.} for binary monotonic weights \cite{dyer2017switch}. However despite these related results, a fast algorithm for sampling general bipartite matchings with monotonic weights remains an interesting open problem, and the application to practical no-regret algorithms for dueling games gives yet another motivation for its study.

\section{Playing Games via DP-Sampling}
\label{sec:dp-sample}

While MCMC-sampling is a powerful tool, standard techniques like Glauber Dynamics may not perform well in settings that require global coordination across coordinates. In this section, we develop a new sampling technique toward this end based on Dynamic Programming, taking advantage of the fact that many settings of interest, such as Colonel Blotto, additionally exhibit certain recursive structure. In particular, we consider a large class of problems called \emph{Resource Allocation Games} that broadly generalize the Colonel Blotto game.

\paragraph{Resource Allocation Game.}
In a resource allocation game,
each player assigns fungible items to some number of battlefields. 
Namely, for the $i$-th player, the action space $\Actions_i $ is the set of ordered size $k_i$ partition of $n_i$ for $k_i, n_i \in \Z^+$. \footnote{While we define actions as all $k$-partitions of $[n]$, our algorithm also applies to action spaces that are subsets of these partitions with arbitrary assignment constraints on each battlefield (i.e., of the form, at most $m$ items can be assigned to battle $j$).}
One can see the action space is indeed a hypergraph where the vertices correspond to pairs $(h, x)$ interpreted as ``assigning $x$ items to the $h$-th battlefield'', and a strategy is simply a subset of vertices $(1, x_1), \cdots, (k, x_k)$ satisfying $\sum_{h=1}^k x_h = n$. 
Let the variables $x_{i, h}$ denote the number of items assigned by the $i$-th player to the $h$-th battlefields 
and $\Actions_{-i} \eqdef \Actions_{1} \times \cdots \times \Actions_{i-1} \times \Actions_{i+1} \times \cdots \times \Actions_{m}$ denote the set of action tuples from players other than player $i$. The reward structure of a resource allocation game is defined by a set of battlefield reward functions $r_{i,h}: [0,n] \times \Actions_{-i} \mapsto \mathbb \R$ for each player $i \in [m]$ and battlefield $h \in [k_i]$.
Let $\vect a_{-i} \in \Actions_{-i}$ be the actions picked by players other than player $i$. The total reward received by the $i$-th player is given by the sum of rewards over individual battles:
$$
R_i( x_{i,1}, \cdots, x_{i, h}, \vect a_{-i})  = \sum_{h=1}^{k_i} r_{i,h}(  x_{i,h} , \vect a_{-i} ).
$$
Additionally, let $r_{i,h,\vect a_{-i}}: [0,n] \mapsto \mathbb \R$ be the restriction of $r_{i,h}$ after fixing the strategies of the other players i.e. $r_{i,h,\vect a_{-i}}(x) = r_{i,h}(x, \vect a_{-i})$. 
We say the resource allocation game is $q$-piecewise and monotonic if for every $i \in [m], h \in [k_i]$, and $\vect a_{-i} \in \Actions_{-i}$, $r_{i,h,\vect a_{-i}}$ is a $q$-piecewise constant and monotonically increasing function.

In the remainder of this section, we discuss how one achieves no-regret learning for the first player (the algorithm for the other players is analogous). For this purpose, we will drop the subscript indicating the player number and let $n = n_1, k = k_1, r_{h} = r_{1,h}, \vect a = \vect a_{-1}$.


\paragraph{RWM Distributions.}
To achieve no-regret learning, we will need to approximately sample from the distributions arising from running the RWM algorithm on Resource Allocation Games.
Assume that the players have played the game for $T$ rounds.
Let $ a^{(t)} \in A_{-1}$ be the action tuples observed from all but the first player at the $t$-th round.
For an action $x \in \Actions_1$ that assigns $x_h$ items to the $h$-th battlefield, RWM will set its weight $w_T(x)$ as 
\begin{align} \label{eq:mwu-weight}
w_T(x) = \beta^{ 
\sum_{h=1}^k \sum_{t=1}^T r_h(x_h,  a^{(t)})}\,,
\end{align}
where $\beta \in [1/2, 1)$ is the learning rate. 
For simplicity, we will define the cumulative battlefield reward function (negated to simplify the syntax appearing after)
\begin{align}
\label{eq:blotto-loss-def}
\ell_h^{(T+1)}(x_h) = - \sum_{t=1}^T r_j(x_h;  a^{(t)}),
\end{align}
so the weight for strategy $x$ can alternatively be written as
$
\prod_{h=1}^k 
\beta^{ \ell_h^{(t+1)}(x_h) }. 
$
We remark that, if the resource allocation game is monotonic and $q$ piecewise, $\ell_j^{(T+1)}(\cdot)$ is also monotonically increasing and a $ ((q-1) \cdot t + 1)$-piecewise constant function. 
Though the domain of $\ell_j$ is of size $n+1$, the property allows us to represent it succinctly in space $O(qt)$,
which is critical when we try to design algorithms whose runtime depends on $n$ polylogarithmically.
In this section, we focus on how one could design algorithms to efficiently sample from $\HD{T}{\beta}$ approximately
given succinct descriptions of the functions
$ r_{h, \vect a^{(t)}}$ for all $h \in [k]$ and $t \in T$ (recall that $r_{h, \vect a^{(t)}}$ is the restriction of the battlefield reward function $r_h$ after fixing the other players' actions). This allows us to implement $\delta$-RWM in polylogarithmic time.


\begin{theorem}[RWM in Resource Allocation Game]\label{thm:resource-allocation-regret}
Let $\mathcal{I} = \GAME$ be an $m$-player monotonic, $q$-piecewise resource allocation game where $\Actions_1 = P_{k}(n)$ for $n, k \in \mathbb Z^+$. 
Suppose the reward of the first player is bounded by $\Lmax$. Then it is possible to implement $\delta$-RWM$^T_\beta$ in time
\begin{align*}
O(Tk) \cdot\Bigg(
&  \min \bigg(
T^2 \Lmax \log(1/\beta)/\delta\cdot 
\zeta_1 \lp( \log T + \log \zeta_2 \rp) \cdot \log n \, ,
n \log n
\bigg)
+ \min(Tq, n)
\Bigg).
\end{align*}
where $\zeta_1 \eqdef \min \lp(\Lmax \log(1/\beta)/ \delta, q  \rp), \zeta_2 \eqdef \max \lp(  \Lmax \log(1/\beta)/ \delta, q  \rp)$, 
assuming access to a $Tq$-piecewise succinct description of $\ell_h^{(t)}$ defined in \Cref{eq:blotto-loss-def}
for all $h \in [k], t \in [T]$.
\end{theorem}
As corollaries of the above theorem, we obtain efficient no-regret learners and algorithms for computing CCEs in Resource Allocation Games. 

In the remainder of this section, we present the sampling algorithm whose analysis leads to the proof of Theorem~\ref{thm:resource-allocation-regret}, and discuss a number of applications to games in the literature including Colonel Blotto and its variants.


\subsection{Sampling via estimation of partition function}

We will focus for the moment on how one could sample from the RWM distribution just in round $T$. For that purpose, we often omit the superscript $T$ for the function $\ell_h^{(T)}$ (Equation~\eqref{eq:blotto-loss-def}) for simplicity. To sample with a dynamic program, we define the partition function for an RWM distribution in resource allocation.

\paragraph{Partition Function}
The partition function $f_h: [0] \cup \Z^+ \mapsto \R^+$ for $h \in [k]$ is defined as the sum of partial weights of all strategies that allocate $y$ soldiers in the subgame induced on the first $h$ battles. Namely,
\begin{align}
\label{eq:partition-function}
    f_h(y)
    = \sum_{ x_1 + \cdots + x_h = y } \prod_{i=1}^h \beta^{ \ell_i(x_i) }.
\end{align}
It has long been known that efficient algorithms for computing the partition function of a self-reducible problem imply efficient (approximate) samplers for the problem's solution space \cite{jerrum1986random}.
As one can see, computing the partition function $f_h$ in our setting simply corresponds to counting the number of (weighted) size $h$ partitions of $y$, which is exactly such a self-reducible problem. Consequently, if one has the values for the partition functions precomputed, one can use them 
to sample from the RWM distribution efficiently. 
We provide the detailed sampling procedure below for completeness.

In particular, this is done by sampling the number of items to put in each battlefield sequentially, conditioned appropriately on prior choices. One puts $x_1 \in \{0, \ldots, n\}$ soldier to the first battlefield with probability
\begin{align} \label{eq:partition-sample-base}
\Pr[x_1 = y] \propto
\beta^{ \ell_1(y) } \cdot f_{k-1}(n - y).
\end{align}
To sample from the $(h+1)$-th battlefield conditioned on the fact that one has put $x_1,\ldots,x_{h}$ soldiers in battles $1\ldots h$, it is enough to sample according to the distribution
\begin{align}
\label{eq:partition-sample-recursive}
\Pr \lp[x_{h+1} = y|  x_{1 \ldots h} \rp] 
\propto 
\beta^{ \ell_{h+1}(y) } \cdot f_{k-h-1}\lp(n -
\lp( \sum_{j=1}^h x_j \rp) - y\rp)  .
\end{align}

The probabilities defined according to Equations~\eqref{eq:partition-sample-base},~\eqref{eq:partition-sample-recursive} yields exactly the RWM distribution, but computing the partition function exactly can be quite costly. For this purpose, we consider the notion of \emph{$\delta$-approximations} (multiplicative) of functions.
\begin{definition}[$\delta$-approximation]
Given $f: [0,n] \mapsto \R^{+}$ and $\hat f:[0,n] \mapsto \R^{+}$, we say $\hat f$ is a $\delta$-approximation of $f$ if for all $x \in [0,n]$ we have
$$ 
\bigg(1 - \delta \bigg) f(x)
\leq \hat f(x) \leq  
\bigg(1 + \delta  \bigg) f(x).
$$
\end{definition}
\noindent Fortunately, with $\delta/(Ck)$-approximations of the partition functions for some sufficiently large constant $C$, one can still perform $\delta$-\textit{approximate} sampling from the RWM distribution (See proof of \Cref{lem:partition-sample}).

Another important observation for achieving efficient (approximate) sampling is that $\ell_h$, the reward function for each battlefield, is a piece-wise constant function. Hence, further optimization is possible if the approximations used for each partition function is also piece-wise constant (and this is indeed the case as we will see shortly). The pseudo-code for the sampling algorithm that takes as input the succinct descriptions of the (approximate) partition functions and the reward functions is given below.

\begin{algorithm}[H]
\begin{algorithmic}[1]
\Require Succinct descriptions $D_{\hat f_h}$ for $h \in [k]$ ; succinct description $D_{g_h}$ of the function $g_h(x) = \beta^{\ell_h(x)}$.
\State Initialize the number of unused soldiers $u = n$.
\State Initialize an empty assignment description $S$.
\For{$h = 1 \ldots k$}
    \State Compute the succinct description $D_{\kappa_h}$ of the
    function  $\kappa_h: [0,u] \mapsto \R^{+}$ defined as 
    \begin{align} 
    \label{eq:prob-weight}
    \kappa_h(i) \eqdef g_h(i) \cdot \hat f_{k-h}\lp(u - i\rp).        
    \end{align}
    \Comment{$|D_{\kappa_h}| = |D_{g_h}| + |D_{\hat f_h}|$ and $D_{\kappa_h}$ can be computed in time linear with respect to the description length.}
    \item[]
    
    \State \{\emph{Compute intervals}\}
    \For{$(a_i, b_i, y_i) \in D_{\kappa_h}$}
    \State Compute the cumulative weight of constant intervals.
    \begin{align}
    \label{eq:piece-weight}
    \nu_i &\eqdef y_i \cdot (b_i - a_i+1).
    \end{align}
    \EndFor
    
    \State \{\emph{Sample an interval}\}
    \State Sample $j \in 1 \ldots |D_{\kappa_h}|$ according to the weight vector $\nu$.
    \item[]
    
    \State \{\emph{Sample soldiers used in battle $h$}\}
    \State Sample $z_h $ uniformly for $\{a_j, a_j + 1, \ldots b_j-1, b_j\}$ where $[a_j, b_j]$ is the $j$-th constant interval in $D_{\kappa_h}$.
    \item[]
    \State Add $z_h$ to the strategy $S$ description.
    \State $u \gets u - z_h$.
\EndFor
\State \Return $S$.
\end{algorithmic}
\end{algorithm}

\begin{lemma}
\label{lem:partition-sample}
Let $\mathcal{I} = \GAME$ be an $m$-player monotonic, $q$-piecewise resource allocation game where $\Actions_i =P_{k_i}(n_i)$ for $n_i, k_i \in \mathbb Z^+$. 
At the $t$-th round, for each $h \in [k]$, let $\hat f_h$ be $\delta/(2k)$-approximations of the partition function defined in Equation~\eqref{eq:partition-function}, and let $g_h(x) = \beta^{ \ell_h(x) }$. 
Assume one is given the succinct descriptions
$D_{\hat f_h}$ and $D_{g_h}$. 
Then, there exists an algorithm \textbf{Partition-Sampling}
which performs $\delta$-approximate sampling from $\HD{T}{\beta}$
in time
$$ 
k \cdot O \lp( p + \min\lp(  Tq  , n \rp) + \log n \rp)
$$
where $p \eqdef \max |D_{\hat f_h}|$.
\end{lemma}
\begin{proof}
If we were to perform the sampling process with $f_h$ instead of $\hat f_h$, we would get exactly the distribution $\HD{t}{\beta}$. This follows from repeated applications of Bayes' rule and that we are sampling the correct conditional distribution. Namely,
\begin{align*}
    \Pr\sbrac{x_1=y_1,\ldots,x_k=y_k} = \prod_{h=1}^k \Pr\sbrac{x_{h}=y_{h} \mid x_{h'}=y_{h'}, \forall h'<h}
\end{align*}

To make sure that we are performing $\delta$-approximate sampling overall, it suffices if we perform $\delta / k$ approximate sampling for each conditional distribution (each battlefield). 
To perform exact sampling, one needs to compute the weight 
$$
\kappa_h^*(i)
\eqdef 
g_h(i) \cdot f_{k-h-1}(u-i).
$$
Since for every $h\in [k]$ we have $\hat f_{h}$ are $\delta/(2k)$ approximations of $f_{h}$, we also have $\kappa_h$ is $\delta/(2k)$ approximations of $\kappa_h^*$, which implies that
$$
(1 - \delta/(2k)) \sum_{i=0}^n \kappa_h^*(i) \leq \sum_{i=0}^n \kappa_h(i) \leq (1 + \delta/(2k)) \sum_{i=0}^n \kappa_h^*(i).
$$
It then follows that the distribution defined by $\kappa_h^*$ and $\kappa_h$ differs by at most $\delta$ in total variation distance.

To analyze the runtime, we note that for each battlefield, we first compute the succinct description $D_{\kappa_h}$ defined in Equation~\eqref{eq:prob-weight}.
Since it is the point-wise multiplication between $g_h$, which is $\min(Tq, n)$-piecewise constant, and $\hat f_{k-h-1}$, which is $p$-piecewise constant, $\kappa_h$ will be $p + \min(Tq, n)$ piecewise constant.
To construct the succinct description of $D_{\kappa_h}$, one maintains two pointers $a=0,b=u$ and keeps track of the interval from $D_{g_h}$ that $a$ is in and the interval from $D_{\hat f_h}$ that $b$ is in. Then, one shifts $a$ forward and $b$ backward to seek for constant intervals of $D_{\kappa_h}$. It is easy to see the runtime of the construction is linear with respect to $|D_{\kappa_h}|$.
Then, computing $\nu_i$ requires scanning through the succinct description of $D_{\kappa_h}$ once. After that, we first sample from a multinomial distribution with support at most $p$, which takes time $O(p)$. Then, we sample from a uniform distribution with support at most $n$, which takes time $O(\log n)$. Adding everything together then gives our final runtime.
\end{proof}
\subsection{Computing the Partition Function}

We now move to showing how to (approximately) compute the partition function. As a warmup, we will first show how this can be done exactly via dynamic programming. In particular, we want to fill a $k \times n$ table such that the $(h, y)$ entry corresponds to the value $f_h(y)$. 
\begin{proposition}
\label{lem:exact-dp}
The values $f_h(y)$ for all $h \in [k]$ and $y \in [0,n]$ can be computed in time 
$ O\lp(nk\log n \rp)$.
\end{proposition}
\begin{proof}
Notice that we have the following recursion
\begin{align}
    f_h(y) = \sum_{x=0}^y \beta^{\ell_h(x)}
    \cdot f_{h-1} \lp( y - x \rp).
\end{align}
$f_h $ is exactly the convolution of $\beta^{\ell_h(\cdot)}$ and $f_{h-1}$. Using Discrete Fast Fourier Transform, $f_h$ can be evaluated in time $O(n \log n)$ (\cite{brigham1988fast}). Hence, in total, the entire DP table can be filled in time $O(n k \log n)$.

\end{proof}


Next, we will discuss how one can develop faster algorithms when $n$ is substantially larger than $k,T$, and $\Lmax$. Our main technical result is an algorithm to pre-compute the partition functions ``approximately'' whose runtime depends on $n$ \emph{polylogarithmically}.

\begin{proposition}
\label{lem:approx-dp}
There exists an algorithm \textbf{Approx-DP} which constructs $\hat f_1, \cdots, \hat f_k$ such that $\hat f_i$ is a $\delta$-approximation of the partition function $f_i$ pointwisely, and runs in time 
$$
O \lp(
k  T^2 \Lmax \log(1/\beta)/\delta\cdot 
\zeta_1 \lp( \log T + \log \zeta_2 \rp) \cdot \log n  
\rp).
$$
where $\zeta_1 \eqdef \min \lp(  \Lmax  \log(1/\beta) / \delta, q  \rp), \zeta_2 \eqdef \max \lp(  \Lmax  \log(1/\beta) / \delta, q  \rp)$.
\end{proposition}
This seems a bit surprising as there are in total $k \cdot (n+1)$ values that we need to pre-compute ($f_h(y)$ for all $h \in [k]$ and $y \in [0,n]$).
However, notice that we are only interested in computing approximations to these values. And, as each $f_h$ is itself a monotonically increasing function, we can approximate it with a sufficiently simple piece-wise function.
\begin{fact}
\label{clm:monotone-approximate}
Given a monotonically increasing function $f: [0,n] \mapsto \R^+$, it can be $\delta$-approximated by a function that is $d$-piecewise constant where $d = \Theta(  \log( \max_x f(x) / \min_x f(x) ) / \delta)$.
\end{fact}
Hence, the algorithm \textbf{Approx-DP} does not need to output the entire $k \times (n+1)$ tables specifying the partition functions. Rather it can just construct the succinct descriptions of a series of functions $\hat f_1, \cdots , \hat f_k$ such that $\hat f_i$ is a $\delta$-approximation of $f_i$. 
By \Cref{clm:monotone-approximate}, we can indeed find such $\hat f_i$ that are $ \Theta( \log( \beta^{ T \Lmax }  )  / \delta) = \Theta( T\Lmax \log(1/\beta)/ \delta ) $ piecewise constant.
As the first building block of the \textbf{Approx-DP}, we demonstrate the routine which, given query access to an unknown monotonically increasing function, constructs a piecewise constant approximation of the function.
\begin{lemma}
\label{lem:monotone-approximate}
Given a monotonically increasing function $f: [0,n] \mapsto \R^+$ and query access to $f$, there exists an algorithm \textbf{Piecewise-Approximate} which outputs a piecewise function $\hat f$ satisfying that 
\begin{itemize}
    \item $\hat f$ is $d$-piecewise constant for $d = \Theta(  \log( \max_x f(x) / \min_x f(x) ) / \delta)$.
    \item $(1 - \delta) \cdot f(x) \leq \hat f(x) \leq f(x)$.
    \item the algorithm runs in time 
    $O \lp( \log n \cdot d \cdot \qt  \rp)$.
\end{itemize}
where $\qt$ is the cost of making a single query to $f$.
\end{lemma}
\begin{proof}
The algorithm proceeds by iteratively finding the longest interval from a starting point such that the function values at the endpoints are within a $(1+\delta)$-factor of each other, and then letting $\hat f$ on this interval be the constant function given by the value of $f$ on the right endpoint. By the monotonicity of $f$, it is easy to see that $\hat f$ is indeed a $\delta$-approximation of $f$ on this interval. The algorithm then repeats this process starting from the next value on which $\hat f$ is not yet defined and repeats until $\hat f$ has been defined on the entire domain.

Since the function values of $f$ increase by at least a factor of $(1 + \delta)$ between each interval and the last, the total number of intervals is at most $d = \Theta(  \log( \max_x f(x) / \min_x f(x) ) / \delta)$.
\end{proof}

\begin{algorithm}[H]
\caption{Piecewise-Approximate} \label{alg:p-approx}
\begin{algorithmic}[1]
\Require Query access to $f: [0,n] \mapsto \R^+$; Approximation accuracy $\delta$.
\State $D \gets \{\}, a \gets 0$.
\While{$a \leq n$}
    \State Binary Search to find the largest $b$ such that $f(b) \leq f(a) \cdot (1 + \delta)$.
    \State Add $(a, b, f(a))$ to $D$. (setting $\hat f(x)$ to $f(a)$ for all $a\leq x\leq b.$) 
    \State $a \gets b+1$.
\EndWhile
\State \Return $D$.
\end{algorithmic}
\end{algorithm}
With this in mind the construction of the series of piecewise constant approximation functions $\hat f_1, \cdots, \hat f_k$ becomes clear: one initializes $\hat f_1 = \textbf{Piecewise-Approximate}(f_1)$ and then defines $\hat f_h$ recursively as $\hat f_h = \textbf{Piecewise-Approximate}( f_h'  )$ where $f_h' (y)= \sum_{x=0}^y \hat f_{h-1}(x)\cdot g_h(y-x)$. Recall that the routine \textbf{Piecewise-Approximate} requires query access to the input function. Hence, we need to show how one could implement query access to $f_1$ and $f_h'$ for $h \in [k]$ efficiently (independent of $n$). The former is easy given the succinct description of $f_1$ since $f_1 = \beta^{\ell_1(x)}$ is a $T \cdot q$ piecewise constant function. 
To realize the latter, we use the fact that $f_h'$ is the convolution of $\hat f_h$ and $g_h$ which are both piecewise constant functions. In particular, we give the following routine which efficiently implements query access to convolutions of piecewise constant functions.

\begin{algorithm}[H]
\caption{Convolution-Query}
\label{alg:convolution-query}
\begin{algorithmic}[1]
\Require Succinct descriptions $D_f$,$D_g$ of two piece-wise constant functions $f$, $g$;
Query point $x$. 
\State Preprocess $D_f$ to extend each tuple into the form $(a_i, b_i, y_i, s_i)$ where $s_i \eqdef \sum_{j < i} (b_j - a_j + 1) \cdot y_j$.
\State Let 
\begin{align*}
F(z) = 
\begin{cases}
\sum_{i=0}^z f(i) \text{ for } i \geq 0 \, ,\\
0 \text{ otherwise.}\\
\end{cases}
\end{align*}
\State $res \gets 0$.
\For{each interval $(a_i, b_i, y_i) \in D_g$}
\State $res \mathrel{+}= 
(F(x - a_i) - F(x - b_i - 1))  \cdot y_i$.
\EndFor
\State \Return $res$.
\end{algorithmic}
\end{algorithm}

We note that the pre-processing incurs a one-time cost for each new function $f$ and does not need to be performed for different queries with respect to the same function
$f$. We also note $F(z)$ can be evaluated by first binary search the first index $j$ such that for $(a_j, b_j, y_j, s_j) \in D_f$ we have
$b_j \geq z$. Then use the equality $F(z) = s_j + (z - a_j) \cdot y_j$.

\begin{lemma}
\label{lem:convolution-query}
Assume one is given the succinct description $D_f, D_g$ of two functions $f,g: [0,n] \mapsto \R^+$.
Let $(f \star g)(x) = \sum_{i=0}^x f(i) \cdot g(x - i)$.
Suppose $|D_f| = p_f$ and $|D_g| = p_g$ with $p_f < p_g$.
There exists an algorithm \textbf{Convolution-Query} that 
takes $O(p_f + p_g)$ time to preprocess $D_f, D_g$ and then takes $O \lp( p_f \cdot \log p_g \rp)$ time to return query access $(f \star g)(x)$ for each $x \in [0,n]$.
\end{lemma}
\begin{proof}
In the pre-processing step, for each tuple $(a_i, b_i, y_i) \in D_f$, we add an extra number $s_i$ which denotes the prefix-sum of all elements before the interval $[a_i, b_i]$. This can be done easily by scanning through the tuples of succinct description in order in one pass. Then, if one want to query the prefix sum $F(z)$, one can just find out which interval $z$ falls into by binary search and then computes in constant time with $s_i$.

Then, using the fact that $g$ is piecewise constant, we can rewrite the convolution query $(f \star g)(x)$ as 
$$
\sum_{i=0}^x f(x-i) \cdot g(i)
= \sum_{i=1}^{ |D_g| } y_i \cdot \lp( \sum_{j=x-b_i-1}^{x-a_i} f(j)\rp)
= \sum_{i=1}^{ |D_g| } y_i \cdot \bigg( F(x - a_i) - F(x - b_i - 1) \bigg).
$$
where $(a_i, b_i, y_i)$ are tuples in $D_g$.
Since evaluating each query to $F$ takes at most $O(\log|D_f|)$ time (for binary search), the above expression can be evaluated in time $O\lp( |D_g| \log|D_f| \rp)$.
\end{proof}
We are now ready to present the pseudocode and analysis of \textbf{Approx-DP}, whose analysis then lead to the proof of \Cref{lem:approx-dp}.
\begin{algorithm}[H]
\caption{Approx-DP} 
\begin{algorithmic}[1]
\Require For each $h \in [k]$, succinct description $D_{g_h}$ of the function $g_h(x) = \beta^{ - \ell_h(x)  }$; number of goods $n$; number of battlefields $k$; approximation error $\delta$.
\State Initialize $\hat f_1$ as $g_1$.
$$
D_{\hat f_1} \gets \textbf{Piecewise-Approximate} \lp( g_1, \delta / (4k) \rp).
$$
\For{$h = 2 \cdots k$}
\State Consider the function
\begin{align}
\label{eq:middle-function}
    f_h' (y)
    = \sum_{x=0}^y 
    \hat f_{h-1}(x)
    \cdot 
    g_h(y-x)
    = \sum_{x=0}^y 
    \hat f_{h-1}(x)
    \cdot 
    \beta^{\ell_h( y-x )}.    
\end{align}
\State Set $\hat f_h$ to be the piecewise approximation of $f_h'$.  \label{line:approximate}
\begin{align*}
    D_{ \hat f_h } \gets \textbf{Piecewise-Approximate} \lp( f_h' ,\delta/(4k)\rp).
\end{align*}
\EndFor
\State \Return $D_{\hat f_h}$ for $h \in [k]$.
\end{algorithmic}
\end{algorithm}

We note that the query access to $g_1$ is by simply reading from the succinct description $D_{g_1}$. Further, query access to $f_h'(y)$ is implemented by routine $\textbf{Convolution-Query}(D_{\hat f_{h-1}} ,D_{g_h},y)$ (see the Appendix for its pseudocode and runtime).

\begin{proof}[Proof of \Cref{lem:approx-dp}]
We will show via induction that each $\hat f_h$ is monotonically increasing and $\hat f_h$ an $(h \cdot \delta/k)$-approximation of the original function $f_h$.
Consider the function $f_h'$ defined in Equation~\eqref{eq:middle-function}.
Since both $\hat f_{h-1}(x)$ and $\beta^{- \ell_h( x )}$ are monotonically increasing, their convolution $f_h'$
is also monotonically increasing. 
Besides, by our inductive hypothesis, 
$\hat f_{h-1}$ is an $((h-1) \cdot \delta / k)$-approximation of $f_{h-1}$, implying that $f_h'$ is an $((h-1) \cdot \delta / k)$-approximation of $f_h$.
By Lemma~\ref{lem:monotone-approximate}, $\hat f_h$ is a $(\delta / k)$-approximation of $f_h'$, and consequently a $(h \cdot \delta/k)$-approximation of $f_h$.

The runtime of the algorithm is dominated by the $(k-1)$ times we call the routine \textbf{Piecewise-Approximate} in Line~\ref{line:approximate}.
Notice that $\max_x f_h(x) / \min_x f_h(y)$ can be at most $\beta^{  -\ell(x) }$ where $\ell(x)$ is at most $T \Lmax $.
As we have argued, each $f_h'$ is always a $\delta$-approximation of $f_h$, which implies that $\max_x f_h'(x) / \min_x f_h'(y)$ is at most $\exp(  2\cdot \log(1/\beta) T \Lmax  )$. Therefore, each $\hat f_h$ is a
$O( \log(1/\beta) T \Lmax  )$ constant function.

\textbf{Piecewise-Approximate} uses the routine \textbf{Convolution-Query} as its query access to the input function $f_h'$.
By Lemma~\ref{lem:convolution-query}, 
\textbf{Convolution-Query} incurs a one-time cost of 
$O( T \Lmax \log(1/\beta) / \delta + Tq  )$ to preprocess the succinct descriptions of $g_h(x) = \beta^{\ell_h(x)}$ and  $\hat f_{h-1}$.
Then, each query takes time 
$
O \lp(  T \cdot \zeta_1 \cdot \lp( \log T + \log \zeta_2 \rp)   \rp).
$ 
where $\zeta_1 \eqdef \min \lp(  \Lmax \log(1/\beta) / \delta, q  \rp), \zeta_2 \eqdef \max \lp(  \Lmax \log(1/\beta) / \delta, q  \rp)$.
By Lemma \ref{lem:monotone-approximate},
each call to \textbf{Piecewise-Approximate} then takes time 
$$
O \lp(
T^2 \Lmax \log(1/\beta)/\delta \cdot 
\zeta_1 \lp( \log T + \log \zeta_2 \rp) \cdot \log n  
\rp).
$$ 
The overall runtime just multiplies the entire expression by $k$.
\end{proof}
Combining Lemmas~\ref{lem:partition-sample},\ref{lem:exact-dp}, and \Cref{lem:approx-dp}, we then obtain an efficient algorithm for $\delta$-RWM in Resource Allocation Games, which concludes the proof of \Cref{thm:resource-allocation-regret}.
\subsection{Applications of the meta algorithm}\label{sec:dp-applications}
In this section, we describe the main applications of our sampling algorithm.

\paragraph{Colonel Blotto Game} A well-studied example of the resource allocation game is the Colonel Blotto Game. In the game, $m$ players try to assign $\{n_i\}_{i=1}^m$ troops to $k$ different battlefields. For the $i$-th battlefield, the player who places more soldiers wins the battle and earn a reward of $w_i \in \Z^{+}$ (ties are broken e.g. lexicographically).\footnote{In the zero-sum variant the other player also loses $w_i$.} This can be viewed as a resource allocation game where the reward function $r_j$ is simply the threshold function $r_j(x) = w_j \cdot \mathbbm 1(x > y)$ where $y$ is the maximum number of soldiers placed by the other players.  It is easy to see that $r_j$ is monotonically increasing and $2$-piecewise constant. Hence, \Cref{thm:resource-allocation-regret} immediately gives an efficient no-regret learning algorithm for the Colonel Blotto Game.
\begin{corollary}[Colonel Blotto without Regret]
\label{thm:cb-no-regret}
Let $\mathcal{I}$ be an $m$-player Colonel Blotto Game where the $i$-th player tries to assign $n_i$ soldiers to $k$ battlefields satisfying $n_i \leq n$.
Then there is a no-regret learning algorithm for $\mathcal{I}$ with regret:
\[
\Reg_T \leq O\left(\Lmax \sqrt{T} \cdot \lp( \sqrt{k\log(n)} + \sqrt{\log(1 / \eta)} \rp) \right)
\]
with probability at least $1 - \eta$ that runs in time 
$$
 Tk \cdot O \Bigg(  \min \lp( 
n \log n,  T^2 \Lmax \log(T \Lmax) \log n
\rp) + m \Bigg).
$$
\end{corollary}
\begin{proof}
By \Cref{lem:apx-mwu}, to achieve no-regret learning in $T$ rounds, we simply need to perform $\delta$-approximate sampling from the RWM distributions with learning rate $\log(1/\beta)$, 
where $\delta = \sqrt{ k \log n / T }$ and $\beta = 1 - \sqrt{k \log n  / T}$. For $T \geq C \cdot k \log n $ for a sufficiently large constant,  we have $\log(1 / \beta) = O( \sqrt{k \log n / T} )$. Hence, $\log(1/\beta)/\delta = O(1)$.
To show that we can perform the sampling process efficiently, we will apply \Cref{thm:resource-allocation-regret} with $q = 2$ since Colonel Blotto is a resource allocation game whose reward function is always $2$-piecewise. To do so, we need to construct and maintain the succinct descriptions of the cumulative battlefield reward function $\ell_h^{(t)}$ required by the sampling algorithm.
Let $a_{j,h}^{(t)}$ be the number of soldiers that the $j$-th player assigns to the $h$-th battlefield at the $t$-th round. Then, for the first player, we essentially have
$$
\ell_h^{(t)}(x)
= w_h \cdot \sum_{t'=1}^{t-1} \mathbbm 1 \lp\{ x > 
\max_{j = 2 \cdots m} a_{j,h}^{(t')} \rp \}.
$$ 
After observing the assignments from other players at round $t$, one first computes the maximum $\nu_h^{(t)} = \max_{j = 2 \cdots m} a_{j,h}^{(t)}$, which can be done in $O(m)$ time. After that, one essentially adds $\ell_h^{(t-1)}$ with the threshold function $\mathbbm 1 \lp\{ x > \nu_h^{(t)} \rp \}$, which takes time $O(\min(T,n))$, which is strictly dominated by the sampling time. Overall, the update just adds an additive factor of $O(Tkm)$ in total.
\end{proof}
This immediately gives an algorithm for computing $\eps$-approximate CCEs in the Colonel Blotto Game (or Nash Equilibria when $m=2$). Namely, we simulate no-regret playing for all $m$ players simultaneously for $T = C \cdot k \Lmax^2 \eps^{-2} \log(mn/\eta)$ many rounds.
\begin{corollary}[Equilibrium Computation for Colonel Blotto Games]\label{cor:blotto-ec}
Let $\mathcal{I}$ be an $m$-player Colonel Blotto Game where the $i$-th player tries to assign $n_i$ soldiers to $k$ battlefields satisfying $n_i \leq n$.
There exists an algorithm to compute an $\eps$-approximate CCE (Nash if the game is two-player zero-sum) with probability at least $1 - \eta$ in time
\begin{align*}
m \cdot 
O \Bigg( 
& \min \bigg(
n k^2 \Lmax^2 \eps^{-2} \log(m/\eta) \cdot \log^2 n\, , \,
k^4 \Lmax^7 \eps^{-6}  \log^4(n) \cdot \log^3(m/\eta)
\cdot \log\lp( k \Lmax \eps^{-1} \log(mn/\eta) \rp)
\bigg) \\
& + m k^2 \eps^{-2} \Lmax^2 \log(mn/\eta)
\Bigg).
\end{align*}
\end{corollary}

In the regime where $n$ is far greater than $k, \eps$, and $\Lmax$, this gives an exponential improvement over prior algorithms \cite{ahmadinejad_dehghani_hajiaghayi_lucier_mahini_seddighin_2019} (at the cost of being approximate rather than exact).


\paragraph{Dice Game}
Dice Games are a randomized variant of Blotto first proposed in \cite{de2006optimal} where $m$-players\footnote{The original dice game is defined for two players. This can be nonetheless generalized naturally to an $m$-player setting.}  construct and roll dice, with the highest roller winning the game. These games are another natural instance of resource allocation. More formally, in the dice game each of $m$-playes has a $k_i$-sided die and $n_i$ points to distribute. The $i$-th player builds their die by assigning a number of dots to each face such that the sum is exactly $n_i$ (where $k_i, n_i \in \mathbb Z^+$). Note that the action space $\Actions_i$ is then exactly the set of ordered partitions $P_{k_i}(n_i)$. The rewards are determined by ``rolling'' the $m$ dice simultaneously; the player with the highest roll wins. In other words, for each player a face is selected uniformly at random and independently, and the player with more dots on the chosen faces wins and earns a reward of $1$. The reward function is given by the expected reward of this process.

Let $x_{i,h}$ be the number of dots that the $i$-th player placed on the $h$-th face of her die for $i \in [m]$ and $h \in k_i$, and denote by $X_i$ the random variable representing the number of dots obtained by player $i$ after rolling. The expected utility for the first player is given by
\begin{align*}
\Pr[ X_1 > \max(X_2, \cdots, X_m) ]
&=
\frac{1}{k_1} \sum_{h=1}^{k_1} \Pr[ x_{1, h} > \max(X_2, \cdots, X_m)] \\
&= \frac{1}{k_1} \sum_{h=1}^{k_1} \prod_{i=2}^m \Pr[ x_{1, h} > X_i]. \\
&= \frac{1}{k_1} \sum_{h=1}^{k_1} \prod_{i=2}^m \frac{1}{k_i} \sum_{h'=1}^{k_i} \mathbbm 1 ( x_{1,h} > x_{i,h'} ). 
\end{align*}
Let $k = \max_i k_i$. It is not hard to see that $\frac{1}{k_i} \sum_{h'=1}^{k_i} \mathbbm 1 ( x_{1,h} > x_{i,h'} )$ is an $O(k)$-piecewise monotonic function and $\prod_{i=2}^m \frac{1}{k_i} \sum_{h'=1}^{k_i} \mathbbm 1 ( x_{1,h} > x_{i,h'} )$ is an $O(mk)$ piecewise monotonic function. 
Hence, this is indeed a $O(mk)$-piecewise monotonic resource allocation game. Applying our meta-algorithm immediately gives the following results on no-regret learning and equilibrium computation in dice games.
\begin{corollary}[No-regret Learning in Dice Games]
Let $\mathcal I = \GAME$ be an $m$ player dice game such that $\max_i k_i \leq k$ and $\max n_i \leq n$.
Then there is a no-regret learning algorithm for $\mathcal{I}$ with regret
\[
\Reg_T \leq O\left(\Lmax \sqrt{T} \cdot \lp( \sqrt{k\log(n)} + \sqrt{\log(1 / \eta)} \rp) \right)
\]
with probability at least $1 - \eta$, and runs in time 
$$
Tk \cdot
O \Bigg(
\min \bigg( 
T^2 \cdot \log(Tm) \cdot \log n, n \log n  \bigg)
+ Tkm
\Bigg).
$$
\end{corollary}
\begin{proof}
Similar to the proof of~\Cref{thm:cb-no-regret}, we will apply~\Cref{thm:resource-allocation-regret} with $q \eqdef mk$ and $\Lmax = 1$. The main step is to maintain the succinct description of $\ell_h^{(t)}$. In dice games, the cumulative reward functions $\ell_h^{(t)}$ for all faces are identical and take the form
$$
\ell^{(t)}(x)
=  \sum_{t'=1}^{t-1} \frac{1}{k_1} \prod_{i=2}^m \frac{1}{k_i} \sum_{h=1}^{k_i} \mathbbm 1 ( x > x_{i,h}^{(t')} ). 
$$
After observing the actions $x_{i,h}^{(t)}$ for $i \in [m]$, $h \in [k_i]$ at the $t$-th round, we first compute the functions $ \nu_i^{(t)}(x) =  \frac{1}{k_i} \sum_{h=1}^{k_i} \mathbbm 1 ( x > x_{i,h}^{(t)} )$, which are all at most $k$-piecewise. This requires sorting $x_{i, 1}^{(t)}, \cdots x_{i, k_i}^{(t)}$, which takes time at most $O(k \log k)$. 
As a result, constructing all $\nu_2^{(t)}, \cdots \nu_m^{(t)}$ takes time in total $O(mk\log k)$. 
Then, we will point-wisely multiply all $\nu_i^{(t)}$ together. One can proceed in a divide and conquer manner: in the first pass, multiply together $\nu_i^{(t)}$ in groups of two, in the second pass, multiply together $\nu_i^{(t)}$ in groups of four and continues until all $\nu_i^{(t)}$ are multiplied together. There will be $\log m$ passes and the computation cost for each pass is at most $O(mk)$. Hence, the process takes time $O(mk \log m)$ in total. Then, we add the resulting function to $\ell^{(t-1)}$, which incurs another cost of $O(Tkm)$. Hence, in total, (assuming $\log m  < T$), it takes time $O( Tkm  )$ to update $\ell^{(t)}$ at one round.
\end{proof}
\begin{corollary}[Equilibrium for Dice Games]
Let $\mathcal{I} = \GAME$ be a dice game with $n = \max(n_1,\ldots, n_m)$, and $k = \max(k_1,\ldots,k_m)$. 
There exists an algorithm to compute an $\eps$-CCE (Nash if the game is two-player zero-sum) with probability at least $1 - \eta$ in time
\begin{align*}
m \cdot O \Bigg(
& \min \big(
n k^2 \eps^{-2}  \log(m/\eta) \cdot \log^2 n,
k^4 \eps^{-6} \log^4(n) \cdot \log^3(m/\eta)
\cdot \log(k \eps^{-1} m \log(n / \eta))
\big) \\
+ & m k^4 \eps^{-4} \log^2(mn/\eta)
\Bigg).
\end{align*}
\end{corollary}
\subsection{Multi-resource allocation games}
\label{sec:multi-resource}
A natural generalization of the resource allocation game is when each player has multiple resource types. This occurs naturally in many settings: a Colonel in Blotto, for instance, might have access to multiple unit types including troops, tanks, and planes (this variant was introduced in \cite{behnezhad_dehghani_derakhshan_hajiaghayi_seddighin_2017}). One would expect that the reward functions should vary depending on which types of units the Colonel chooses.

More formally, in the multi-resource allocation game, the $i$-th player has $B_i$ types of fungible items. We denote by $n_{i,b}$ the number of type-$b$ items that the $i$-th player possesses. Her strategy is an allocation of these items to $k_i$ battlefields. We denote $X_{i,b,h}$ as the number of type $b$ items that the $i$-th player assigns to the $h$-th battlefield. Similar to the single-resource allocation game, for each player $i$ and each battlefield $h$, there is a battlefield reward function 
$$
r_{i,h}: ( [0,n_{i,1}] ) \times
\cdots \times ( [0,n_{i,B}] )
\times \Actions_{-i} \mapsto \R \, ,
$$
where we recall $\Actions_{-i}$ is the set of strategy tuples from the players other than $i$. 
Let $ S \in \Actions_{-i}$ be the strategies used by the other players, the total reward for the $i$-th player on strategy $X=\{X_{i,b,h}\}$ is given by summing over the rewards on each individual battlefield:
$$
R_i(X,S) = \sum_{h=1}^{k_i} r_{i,h} \lp(X_{i, 1, h}, \cdots, X_{i, b, h},  S \rp). 
$$
We now prove a variant of \Cref{thm:resource-allocation-regret} for the multi-resource setting.
\begin{theorem}[RWM in Multi-Resource Allocation Game]\label{thm:multiresource-allocation-regret}
Let $\mathcal{I} = \GAME$ be an $m$-player multi-resource allocation game where $\Actions_i = P_{k_i}(n_{i,1}) \times \cdots \times P_{k_i}(n_{i,B})$ for $n_1, \cdots n_B, k \in \mathbb Z^+$.
Suppose the reward of the first player is bounded by $\Lmax$, $\max_{i \in [m], b\in [B]} n_{i,b} = n$, $\max_{i \in [m]} k_i = k$.
Then it is possible to implement (exact) RMW$^T_\beta$ in time
\begin{align*}
O\lp(T k B(n+1)^{2B}  + m T k^2 B  \cdot (n+1)^B  \rp).
\end{align*}
assuming query access to $r_{i,h}: ( [0,n_{i,1}] ) \times
\cdots \times ( [0,n_{i,B}] )
\times \Actions_{-i} \mapsto \R$.
\end{theorem}
\begin{proof}
We proceed to analyze the partition function of the RWM distribution. 
As usual, we will conduct the analysis from the first player's perspective and drop the subscript used to index the player.
Suppose the game is played for $T$ rounds and the observed actions from the other players' are $S^{(1)}, \cdots, S^{(T)}$. 
Then, similar to Equation~\eqref{eq:blotto-loss-def}, for an assignment $x \in \Actions_i$, we will define the cumulative reward  
$$
\ell^{(T+1)}_h(X_{1, h}, \cdots, X_{b, h}) = \sum_{t=1}^T r_{h} \lp(X_{1, h}, \cdots, X_{B, h}, S^{(t)} \rp).
$$
For simplicity, we will abbreviate $(X_{1, h}, \cdots, X_{b, h})$ as $X_{*, h}$.
Then, we have the weight for the action $x$ is simply
$$
w_T(X) = \prod_{h=1}^k 
\beta^{ \ell_h^{(T+1)}(X_{*, h}) }.
$$
After dropping the superscript marking the rounds, accordingly, the partition function is now $f_h: \lp([0,n] \rp)^B \mapsto \R^+$ for $h \in [k]$ defined as
\begin{align*}
    f_h(y)
    = \sum_{ X \in \mathcal X_y } \prod_{h=1}^k \beta^{ \ell_h(X_{*,h}) } \, ,
\end{align*}
where $\mathcal X_y = \{
X \in (\Z^+)^{B \times k} | \forall b \in [B] \, , \sum_{h=1}^k X_{b,h} = y_b
\}$. We still have the recursion
\begin{align}
\label{eq:multi-dim-recursion}
    f_h(y)
    = \sum_{ z \in  {\mathcal Z_y} }  f_h(y -z)
    \cdot 
    \beta^{ \ell_h(z)}
    \, ,   
\end{align}
where $\mathcal Z_y = \{
z \in ([0,n])^{B} | \forall b \, ,  z_{b} \leq y_b \}$.

Compared to the single-resource allocation game, there are 
$O\lp(k \cdot (n+1)^B\rp)$ partition function values we need to compute. Using dynamic programming and the recursion stated in Equation~\eqref{eq:multi-dim-recursion}, each of them now takes time at most $O\lp( B (n+1)^B \rp)$. Hence, filling the entire DP table takes time $O\lp( k B \cdot (n+1)^{2B} \rp)$. 

After that, we likewise sample the assignment for each battlefield sequentially. 
We will write $\vect n = (n_1, \cdots, n_B)$.
For the first battlefield, we sample 
$$
\Pr[ X_{*,1} = y ] \propto \beta^{ \ell_1(y) } \cdot f_{k-1} \lp(  \vect n - y\rp).
$$
To sample from the $(h+1)$-st battlefield, one sample according to the distribution
\begin{align*}
\Pr \lp[X_{*, h+1} = y|  X_{*, 1 \cdots h} \rp] 
\propto 
\beta^{ \ell_{h+1}(y) } \cdot f_{k-h-1}\lp( \vect n -
\lp( \sum_{j=1}^h X_{*, j} \rp) - y\rp).
\end{align*}
The domain size of the distributions we sample from is $O\lp((n+1)^B \rp)$.  
To compute the probabilities of each element takes $O(B)$ times. Hence, the sampling time for one battlefield is
$O\lp(B \cdot (n+1)^B \rp)$. 
Hence, the runtime of the sampling process is dominated by that of computing the partition functions.

Finally, we discuss how we maintain the function $\ell_h^{(t)}$ at round $t$. To do that, we gather the strategy $S^{(t)}$ observed at the $t$-th round from other players and then query
$r_h(z, S^{(t)}) $
for each $z \in \lp([0,n]\rp)^B$, and add that to $\ell_h^{(t-1)}(z)$. Each query takes time $O(mBk)$ (to write down the input). In total, maintaining $\ell_h^{(t)}$ in $T$ rounds takes time $O( mTk^2B (n+1)^B )$. Adding this together with the time for computing the partition function then gives our final runtime.
\end{proof}
As an immediate application, we get no-regret learning and equilibrium computation for the multi-resource allocation games such as the multi-resource Colonel Blotto problem.
\begin{corollary}[Multi-resource without Regret]\label{thm:multiresouce-cb-no-regret}
Let $\mathcal{I} = \GAME$ be an $m$-player multi-resource allocation game where $\Actions_i = P_{k_i}(n_{i,1}) \times \cdots \times P_{k_i}(n_{i,B})$ for $n_1, \cdots n_B, k \in \mathbb Z^+$.
Suppose the reward of the first player is bounded by $\Lmax$, $\max_{i \in [m], b\in [B]} n_{i,b} = n$, $\max_{i \in [m]} k_i = k$.
Then there is a no-regret learning algorithm for $\mathcal{I}$ with regret:
\[
\Reg_T \leq O\left(\Lmax \sqrt{T} \cdot \lp( \sqrt{  \sum_{b=1}^B \log(n_B)} + \sqrt{\log(1 / \eta)} \rp) \right)
\]
with probability at least $1 - \eta$ that runs in time 
$$
O\lp(T k B(n+1)^{2B}  + m T k^2 B  \cdot (n+1)^B  \rp)
$$
assuming query access to $r_{i,h}: ( [0,n_{i,1}] ) \times
\cdots \times ( [0,n_{i,B}] )
\times \Actions_{-i} \mapsto \R$.
\end{corollary}
\begin{corollary}[Equilibrium Computation for Multi-Resource Allocation Games]
Let $\mathcal{I} = \GAME$ be an $m$-player multi-resource allocation game where $\Actions_i = P_{k_i}(n_{i,1}) \times \cdots \times P_{k_i}(n_{i,B})$ for $n_1, \cdots n_B, k \in \mathbb Z^+$.
Suppose the reward of the first player is bounded by $\Lmax$, $\max_{i \in [m], b\in [B]} n_{i,b} = n$, $\max_{i \in [m]} k_i = k$.
There exists an algorithm to compute an $\eps$-approximate CCE (Nash if two-player zero-sum) with probability at least $1 - \eta$ in time
$$
O\lp(m (n+1)^{2B} k^2 B^2 \Lmax^2\eps^{-2} \log (mn/\eta)  + m^2 (n+1)^B k^3 B^2 \Lmax^2 \eps^{-2} \cdot  \log (mn/\eta)     \rp)
$$
assuming query access to $r_{i,h}: ( [0,n_{i,1}] ) \times
\cdots \times ( [0,n_{i,B}] )
\times \Actions_{-i} \mapsto \R$.
\end{corollary}



\section*{Acknowledgements}
We thank Thuy Duong Vuong, Nima Anari, and Kuikui Liu for helpful discussions on MCMC-based approaches for Dueling games and Colonel Blotto, and Saeed Seddighin for discussions on algorithms for Colonel Blotto and related problems. We also thank Maxwell Fishelson for enlightening discussion of optimistic hedge and its variants.

\bibliographystyle{amsalpha}  
\bibliography{references} 
\appendix
\section{No-regret Learning and Equilibrium Computation}\label{app:nr}
\begin{lemma}[Lemma 4.1 from \cite{cesa2006prediction} Rephrased]
\label{lem:high-prob-regret}
Let $\mathcal I = \GAME$ be an $m$-player game played repeatedly for $T$ rounds. 
Denote $\vect s^{(t)}_i$ as the mixed strategy chosen by the $i$-th player, and $a^{(t)}_i \sim \vect s^{(t)}_i$ as the action sampled.
Assume the $i$-th player follows an algorithm which computes $\vect s^{(t)}_i$ solely based on $a^{(t')}_j$ for $t' < t$ and $j \in [m] \backslash \{i\}$. Furthermore, suppose the following is true
$$
\sup_{ b^{(t)}_j \in \Actions_j \text{ for } t \in [T], j \neq i  } 
\E \lp[ 
\max_{ e \in \Actions_i }
\sum_{t=1}^T R_i(e, b^{(t)}_{-i})
- \sum_{t=1}^T R_i( a^{(t)}_i, b^{(t)}_{-i} )
\rp]
\leq B.
$$
Then, for all $\delta \in (0,1)$, with probability at least $1 - \delta$, it holds
$$
\max_{ e \in \Actions_i }
\sum_{t=1}^T R_i(e, a^{(t)}_{-i})
-\sum_{t=1}^T R_i( a^{(t)}_i, a^{(t)}_{-i} )
\leq B + \Lmax \sqrt{ T/2 \log(1/\delta) }.
$$
at the end of the repeated play for the $i$-th player.
\end{lemma}

\begin{proof}[Proof of Lemma~\ref{lem:apx-mwu}]
If one follows exactly from the Randomized Weighted Majority algorithm, one has the guarantee that 
$$
\sup_{ b^{(t)}_j \in \Actions_j \text{ for } t \in [T], j \neq i  } 
\E_{ a^{(t)}_i \sim \HDI{t}{i}{\beta} } \lp[ 
\max_{ e \in \Actions_i }
\sum_{t=1}^T R_i(e, b^{(t)}_{-i})
- \sum_{t=1}^T R_i( a^{(t)}_i, b^{(t)}_{-i} )
\rp]
\leq \Lmax \sqrt{ T \log N }.
$$
In reality, since we are performing $\kappa$-approximate sampling, we have $\TV \lp( \HDI{t}{i}{\beta}, \vect s^{(t)}_i \rp) \leq \kappa$. Since the reward of the game is bounded by $\Lmax$, we have
$$
\sup_{ b^{(t)}_j \in \Actions_j \text{ for } t \in [T], j \neq i  } 
\E_{ a^{(t)}_i \sim \vect s^{(t)}_i } \lp[ 
\max_{ e \in \Actions_i }
\sum_{t=1}^T R_i(e, b^{(t)}_{-i})
- \sum_{t=1}^T R_i( a^{(t)}_i, b^{(t)}_{-i} )
\rp]
\leq \Lmax \sqrt{ T \log N } + \kappa \cdot \Lmax \cdot T.
$$
Then, by Lemma~\ref{lem:high-prob-regret}, it then holds
$$
\frac{1}{T} 
\max_{ e \in \Actions_i }
\sum_{t=1}^T R_i(e, a^{(t)}_{-i})
- \frac{1}{T} \sum_{t=1}^T R_i( a^{(t)}_i, a^{(t)}_{-i} )
\leq \Lmax \sqrt{ \log N / T } + \kappa \cdot \Lmax  + \Lmax \sqrt{ \log(1/\delta) / (2T) }.
$$
with probability at least $1 - \delta$.
Setting $T = \Lmax^2 \eps^{-2} \log (N / \delta)  $ and $\kappa = \min(1/2,  )$ then gives the average regret is bounded by $O(\eps)$.
\end{proof}

\begin{proof}[Proof of \Cref{cor:compute-cce}]
By Lemma~\ref{lem:apx-mwu}, if we set $T = C \cdot \lp( \Lmax^2 \eps^{-2} \log (Nm/\eta)\rp)$, $\delta = \eps / (C \Lmax)$ for a sufficient large constant and simulate the repeated game playing for $T$ rounds where each player makes her decision based on $\delta$-RWM, the regret of the $i$-th player is bounded by $\eps$ with probability at least $1 - \delta/m$. By union bound, this holds for all players simultaneously. Ths results then follows from \Cref{thm:no-regret-cce}.
\end{proof}
\section{Bit-complexity and Stability of Numeric Operations}\label{app:bit}
To avoid un-necessary technical details of the bit-complexities of numbers and time complexities of algebraic operations, the algorithmic results in the main body of this work are stated in the \emph{Algebraic Computation Model}. In particular, we assume additions, subtractions, multiplications, divisions, exponentiation and comparisons can be carried between real numbers in constant time, and that the computing device has query access to the digits of the real numbers. 
We remark that, for all the games we study, if the rewards are rational numbers with bounded bit complexities, our algorithms can all be implemented exactly in the RAM computation model with their runtime increased by at most polynomial factors. 
Unsurprisingly, if one is more careful with the numeric precision needed and maintains only multiplicative approximations of each algebraic operations, our algorithms can be implemented in the RAM model losing only poly-logarithmic factors. In this section, we discuss some standard techniques to this end for the reader interested in any implementation of our algorithms.

We first discuss how we represent and perform algebraic operations on numbers whose absolute values are exponentially large or small. 
While writing these numbers down exactly is costly, for the purpose of $\delta$-approximate sampling, it is actually sufficient t keep $\poly(\delta)$ multiplicative approximations of these numbers. Fortunately, these approximations can indeed be represented much more succinctly using the scientific notations. For convenience, for a number $a \in \R^+$, we will call $\tilde a$ a $\delta$-approximation of $a$ if we have $ (1-\delta)a \leq \tilde a \leq (1 + \delta)a  $ and a one-sided $\delta$-approximation if we have
$(1-\delta)a \leq \tilde a \leq a  $.
\begin{fact}
Given $a \in \R^+$ satisfying $\exp(-q) \leq a \leq \exp(q)$ for $q \in \Z^+$, let $\tilde a$ be $a$ written in scientific form keeping $ \Theta ( \log(1 / \delta) )$ many significant figures. Then, $\tilde a$ is a
one-sided $\delta$-approximation of $a$ and can be represented using 
$\Theta(  \log(\delta^{-1}) + \log q )$ many bits.
\end{fact}
Instead of performing exact arithmetic computations, we can perform `approximate' arithmetic operations on real numbers in all our algorithms.
\begin{claim} 
\label{clm:efficient-arithmetic}
Let $a,b \in \R^+$ be two numbers in scientific notations with $s$ significant figures.
\begin{itemize}
    \item one-sided $\delta$-approximation of additions and multiplications can be performed in time $ O(s)$ and $O(s \log s)$ respectively.
    \item $\delta$-approximation of division can be performed in time 
    $O( s + \log(1 / \delta) )$.
    \item Given $1 \leq \alpha \leq 2$ that has $s$ significant figures and $i \in \mathbb Z^+ \cup \{0\}$, a one-side
  $\delta$-approximation of $\alpha^{i}$ can be computed in time $ O( s \cdot \log^2 i \cdot \log(\delta^{-1}) )$.
\end{itemize}
\begin{proof}
The first two claims follow from the definition of (one-sided) $\delta$-approximation. We proceed to show that one can perform approximate exponentiation efficiently.
In particular, we argue that $\alpha^{i}$ can be computed fairly accurately via fast exponentiation while keeping  $C \cdot \log( i/\delta)$ significant figures throughout the computation for some large enough constant $C$. 
By doing so, we can make sure the approximation to 
$\alpha$ is a one-sided $\xi$-approximation  where $\xi = \delta / i^c$ for some large enough constant $c$.
Consequently, the approximation of $\beta^{j}$ for any $j \in [i]$ that is a power of $2$ is within $(1 \pm 4^{ \log_2(i) } \cdot \xi)$. 
It then follows $\alpha^i$ can be approximated within $(1 \pm \xi \cdot O(\log i) \cdot 4^{ \log_2(i) }) = (1 \pm \delta)$ when $c$ is sufficiently large.
\end{proof}
\end{claim}
Unsurprisingly, the output of applying a series of arithmetic operations will be within multiplicative factors of the result obtained by replacing each operation with its approximate counterpart.
\begin{fact}
\label{clm:add-mul-error}
Given a variable $y$ that is the result of $V$ arithmetic operations including Addition, Multiplication and Division on the inputs $x_1, \cdots, x_n \in \R^+$ in scientific notations with $s$ significant figures, let $\hat y$ be the variable obtained by replacing all the arithmetic operations with their $\delta/(10V)$-approximate counterparts for small enough $\delta$. Then, $\hat y$ will be a $\delta$-approximation of $y$. Moreover, if only additions and multiplications are used, the approximation is one-sided i.e. $\hat y \leq y$.
\end{fact}
Careful readers may find that subtraction is excluded when we discuss approximate algebraic operations in Claim~\ref{clm:efficient-arithmetic}. 
For two numbers $a,b \in \R^+$ and $\hat a, \hat b$ be their $\delta$ approximations counterparts, $\hat a - \hat b$ may be wildly different from $a - b$ when $a$ is substantially larger than $b$. 
Yet, subtraction between real numbers is indeed used in two different places. 
Firstly, subtractions occur in Discrete Fourier Transform, which is used in \Cref{lem:approx-dp} to compute the convolution between functions. 
The numeric stability of DFT varies among different implementations and depends on a number of subtle factors (See \cite{schatzman1996accuracy}).
If numeric stability indeed becomes an issue in the actual implementation, one can fallback to evaluate the convolution in the brute-force manner, which increases the complexity from $O(nk\log n)$ to $O(n^2 k)$.

Another place where subtractions are used is in Algorithm~\ref{alg:convolution-query} to compute range sum of piece-wise function efficiently. As such, we need a numerically more stable technique for performing range sum query in place of the prefix sum technique. 
In particular, given the succinct description of a $q$-piecewise constant function $f: \{0\} \cup [n] \mapsto \R^+$ and $\delta \in (0,1)$ beforehand, we want to perform some preprocessing in time $q \cdot \polylog(n, q, \delta^{-1})$ and then answer a series of queries of the form $\sum_{i=a}^b f(i)$ within $(1 \pm \delta)$ multiplicative factors in time $\polylog(n, q, \delta^{-1})$.
\begin{claim}
\label{clm:range-sum}
Given the succinct description $D_{f}$ of a
$q$-piecewise constant function $f: \{0\} \cup [n] \mapsto \R$ where the function values contain at most $s$ significant figures, there exists an algorithm \textbf{Range-Sum-Query} which performs some preprocess in time $O(qs) + \polylog(q,\delta^{-1}, n)$, and can compute one-sided $\delta$-approximation to query of the form $\sum_{i=a}^b f(i)$ in time $\polylog(n, q, \delta^{-1})$.
\end{claim}
\begin{proof}
Let $D_{f} = \{ (a_1, b_1, y_1), \cdots (a_q, b_q, y_q)  \}$.
In the preprocessing step, we first compute the range sum of all intervals $[a_i, b_i]$ (approximately). Denote the results as an array $[s_1, \cdots, s_q]$. Then, we build a segment tree with the array, where nodes store the approximate range sum of intervals of lengths that are powers of $2$. This takes  $O(q \log(q))$ arithmetic operations. The data structure then allows us to answer $\sum_{i=a}^{b} s_i$ with $ O(\log(q))$ arithmetic operations with enough accuracy.
Then, when we receives a query $\sum_{i=a}^b f(i)$. We first binary search for the intervals of $f$ that $a,b$ fall into respectively. Next, one uses the pre-built segment tree to answer the range sum of any intervals that are strictly contained in $[a,b]$ and then adds the sum of remaining elements. 
It is then not hard to see that the pre-processing step takes time at most $O(qs) + \polylog(q,\delta^{-1}, n)$, and answering each query takes time $\polylog(q,\delta^{-1}, n)$.
\end{proof}

Lastly, we discuss the building block of sampling: 
sampling from multinomial distributions.
Typically, our algorithm computes a vector $w_1, \cdots, w_n$ and then samples from the multinomial distribution $X$ where $\Pr[X = i] \propto w_i$. This is simple in the Algebraic computation model as one can easily reduce this to sampling from uniform distributions over real intervals (which can be done at assumed unit cost). In particular, one first computes the prefix sum $W_1, \cdots W_n$. Then, one samples $z$ from the uniform distribution over the interval $[0, W_n]$ and returns $j$ for $ W_{j-1} < z < W_{j}$. This clearly takes at most time $O(n)$. 
In the bit-complexity model, we can nonetheless achieve approximate sampling from arbitrary multinomial distributions with similar runtime.
\begin{claim}
\label{clm:efficient-sample}
 Given a weight vector $(w_1, \cdots, w_n)$ in scientific notations with $s$ significant figures, $\delta$-approximate sampling from the multinomial distribution $X$ such that $\Pr[X = i] \propto w_i$ can be done in time $\wt O( n \cdot ( \log (\delta^{-1}) + s ) )$.
\end{claim}
\begin{proof}
After reading the input, one first truncates to make sure each $w_i$ has at most $O(\log(n/\delta))$ many significant figures as that is already enough for the specified `sampling accuracy'. After that, all arithmetic operations will be carried out with their $c \cdot \delta/n$ approximations for some sufficiently small constant $c$.
In the next step, one normalizes the weight vector and rounds each $w_i$ to their nearest multiple of $c \cdot \delta / n$. Doing so changes the distribution by at most $c \cdot \delta$ in total variation distance. One can then multiply all $w_i$ by a factor of $n / (c \cdot \delta)$ to make everything an integer. Finally, one can do the same thing as sampling in the Algebraic computation model: computing the prefix sums and reducing the problem to sampling from uniform distributions, now over integer intervals. The integers in the interval can be at most $ n^2 / (c \cdot \delta)$ so the runtime is dominated by the preliminary computations performed.
\end{proof}

We note that in many cases we actually require a slightly more complicated sampling procedure where we wish to sample from a $q$-piecewise support-$n$ multinomial with $q \ll n$. This can be done similarly in time $\wt O((q+\log(n))( \log (\delta^{-1}) + s ))$ by first sampling one of the $q$ piecewise intervals by the above technique, then sampling uniformly within the interval.
\section{Implementing Glauber Dynamics}\label{app:Glauber-implementation}

This section is devoted to proving \Cref{thm:matroid-no-regret}, which we repeat here for convenience.
\begin{theorem}[RWM on Matroids]
Let $\mathcal{I}=\GAME$ be an $m$-player game on a size-$n$ ground set $\Omega$. If $A_i$ consists of the bases of a rank-$k$ matroid, is linear, and is collision-sensitive with support $q$, then it is possible to implement $\delta$-RWM$^T_\beta$ in time
\[
O\left(k_iT(CO + q\log(n)+mk_{-i}T\log(n)) \log\left( \frac{k_i\log(n) + \Lmax T\log(\beta^{-1})}{\delta}\right)\right),
\]
assuming access to a $q$-piecewise succinct description of $\NC_i$ encoded under an ordering of $\Omega$ and a contraction oracle matching the same ordering.
\end{theorem}
\begin{proof}
Since $A_i$ is given by the bases of a rank-$k$ matroid, the single-step Glauber Dynamics on $\Omega^w$ mix in time
\[
T(GD(\Omega^w), \delta) \leq O\left(k\log \left( \frac{\log(|A_i|/w_*)}{\delta}\right)\right)
\]
for any external field $w \in \R_+^\Omega$. In our setting, we have that $|A_i| \leq n^k$, and $w_* \geq \beta^{\Lmax T}$, and the process needs to be repeated once per round bringing the complexity to:
\[
O\left(kT \log\left( \frac{k\log(n) + \Lmax T\log(\beta^{-1})}{\delta}\right)\right)
\]
times the implementation time of a single step of Glauber dynamics.

It is left to bound this cost. To implement a step of GD in the $t$-th round of optimistic hedge, we first remove a uniformly random element from our current basis, then re-sample from the conditional distribution. The first step can easily be implemented in $O(\log(k))$ time. The latter step requires more care. Let $\hat{e}$ denote the $(k-1)$-size set resulting from the down-step of the walk. Query the contraction oracle on $\hat{e}$ and call the resulting set $S_e \subset E$. Notice that by definition, the conditional measure of any $x \in S_e$ is proportional to $\beta^{-r_t(x)}$ where
\[
r_t(x) = \sum\limits_{j=1}^t R_i^\Omega(x,s^{(j)}).
\]
Thus to perform the conditional sampling efficiently, it is sufficient to compute the external field for each element in $S_e$ and sample from the corresponding multinomial distribution.

While implementing this naively would require time at least $|S_e|$ to check the weight of each element in the conditional distribution, this can be circumvented via our assumption that our game is collision-sensitive with bounded support. In particular, assume for the moment we have access to a succinct description for the vertex-wise total rewards $r_t(v)$ that is $(q+tmk_{-i})$-piecewise, and that the output of the contraction oracle respects the order of the description (we will argue this can be constructed efficiently shortly). As a result, the total rewards in $S_e$ are $(q+tm k_{-i})$-piecewise as well. This means that using query access to CO,\footnote{Formally we are also assuming here one has query access to the size of the output of the contraction oracle. Note this can be easily implemented in $\text{polylog}(n)$ time even if one does not assume such access.} we can build a succinct description for total rewards on the elements in $|S_e|$ (labeled by their index in CO). Sampling from the corresponding multinomial distribution in the algebraic computation model then takes $O((q+tm k_{-i}))$ time, and one can then feed the sampled index into CO to receive the correct vertex. Altogether, a single step of GD can therefore be implemented in $O(CO + (q+tm k_{-i})\log(n))$ time assuming access to the appropriate description of total rewards.

Finally, we argue we can construct and maintain the succinct descriptions of the vertex-wise reward functions over $T$ rounds efficiently. 
Recall we start with an $q$-piecewise succinct description for the no-collision vertex-wise reward values. In each round, at most $m k_{-i}$ new elements of $\Omega$ are introduced into the history, and since the game is collision-sensitive the resulting succinct description of rewards is at most $(q+tm k_{-i})$-piecewise in the $t$-th round as desired. The computational cost stems from noting that it is actually sufficient just to update the rewards for vertices which have appeared in the opponent history (and the number of rounds in which it has appeared). During look-up, computing the total reward for any vertex $v$ that has appeared $t$ times can be computed in $O(1)$ time by simply adding the stored value $(T-t)R_i(v,s)$ for any $s\notni v$. The cost of building the succinct description is therefore asymptotically dominated by the sampling procedure above, which gives the final complexity.
\end{proof}


\end{document}